\documentclass[conference]{llncs2e/llncs}
\usepackage{etex}

\usepackage[textfont={it,scriptsize},reqno]{subfig}

\usepackage{mathtools} 
\usepackage[noend]{algpseudocode}
\usepackage{algorithm}
\usepackage{cprotect}
\usepackage{cases}
\usepackage{booktabs}
\usepackage{llncs2e/IEEEtrantools}

\usepackage{relsize}
\usepackage{amsmath,amssymb,mathtools,stmaryrd}
\usepackage{graphicx}

\usepackage{complexity}
\usepackage{tikz}
\usepackage{xspace}
\usepackage{multicol}
\usepackage{wrapfig}
\usepackage{enumitem}
\usepackage[bottom]{footmisc}
\usepackage{footnote}
\makesavenoteenv{table}
\makesavenoteenv{tabular}
\usepackage[font={scriptsize,it}]{caption}
\usepackage{adjustbox}
\usepackage{color}
\usepackage{cancel}
\usepackage[normalem]{ulem}

\setlist[itemize]{%
labelsep=8pt,%                           
labelindent=0\parindent,%               
itemindent=0pt,%
leftmargin=*,%                          
listparindent=-\leftmargin,% 
nolistsep
}

\setlist[enumerate]{%
labelsep=8pt,%                           
labelindent=0\parindent,%               
itemindent=0pt,%
leftmargin=*,%                          
listparindent=-\leftmargin,%
nolistsep 
}

\setlist[description]{%
labelsep=8pt,%                           
labelindent=0\parindent,%               
itemindent=0pt,%
leftmargin=0pt,%                          
listparindent=\parindent% 
}

\newlist{provisos}{enumerate}{1}     % this creates a dedicated counter named 'subtaski'
\setlist[provisos,1]{
itemsep=0pt,
label*={\textbf{D\arabic*}},
ref={D}\arabic*,start=0,              
listparindent=-\leftmargin,%          
itemindent=0pt,%
leftmargin=*,% 
topsep=4pt,
}
\newlist{provisosp}{enumerate}{1}     % this creates a dedicated counter named 'subtaski'
\setlist[provisosp,1]{
itemsep=0pt,
label*={\textbf{D\arabic*'}},
ref={D}\arabic*',start=2,              
listparindent=-\leftmargin,%          
itemindent=0pt,%
leftmargin=*,% 
topsep=4pt,
}

\newlist{lipton}{enumerate}{1}     % this creates a dedicated counter named 'subtaski'
\setlist[lipton,1]{
itemsep=0pt,
label*={\textbf{L\arabic*}},
ref={L}\arabic*,
start=1,
listparindent=-\leftmargin,%          
itemindent=0pt,%
leftmargin=*,% 
topsep=4pt,
}

\algrenewcommand\algorithmicindent{1.0em}%

\usepackage{ifpdf}
    \ifpdf
\usepackage{hyperref}
    \else
    \fi

\usetikzlibrary{automata}
\usetikzlibrary{positioning}
\usetikzlibrary{arrows,patterns}
\usetikzlibrary{calc}
\usetikzlibrary{decorations.pathreplacing,decorations.markings}
\usetikzlibrary{decorations.pathmorphing}

\renewcommand\paragraph[1]{\vspace{.0ex}\par\noindent\textbf{#1}.\;\;}

\newtheorem{proof}{Proof}
\let\phi=\varphi 
\let\epsilon=\varepsilon

\renewcommand\implies{\Rightarrow}

\makeatletter
\newcommand*{\da@rightarrow}{\mathchar"0\hexnumber@\symAMSa 4B }
\newcommand*{\da@leftarrow}{\mathchar"0\hexnumber@\symAMSa 4C }
\newcommand*{\xdashrightarrow}[2][]{%
  \mathrel{%
    \mathpalette{\da@xarrow{#1}{#2}{}\da@rightarrow{\,}{}}{}%
  }%
}
\newcommand{\xdashleftarrow}[2][]{%
  \mathrel{%
    \mathpalette{\da@xarrow{#1}{#2}\da@leftarrow{}{}{\,}}{}%
  }%
}
\newcommand*{\da@xarrow}[7]{%
  \sbox0{$\ifx#7\scriptstyle\scriptscriptstyle\else\scriptstyle\fi#5#1#6\m@th$}%
  \sbox2{$\ifx#7\scriptstyle\scriptscriptstyle\else\scriptstyle\fi#5#2#6\m@th$}%
  \sbox4{$#7\dabar@\m@th$}%
  \dimen@=\wd0 %
  \ifdim\wd2 >\dimen@
    \dimen@=\wd2 %   
  \fi
  \count@=2 %
  \def\da@bars{\dabar@\dabar@}%
  \@whiledim\count@\wd4<\dimen@\do{%
    \advance\count@\@ne
    \expandafter\def\expandafter\da@bars\expandafter{%
      \da@bars
      \dabar@ 
    }%
  }%  
  \mathrel{#3}%
  \mathrel{%   
    \mathop{\da@bars}\limits
    \ifx\\#1\\%
    \else
      _{\copy0}%
    \fi
    \ifx\\#2\\%
    \else
      ^{\copy2}%
    \fi
  }%   
  \mathrel{#4}%
}
\makeatother

\newcommand\defmath[2]{\newcommand#1{\ensuremath{#2}\xspace}}
\newcommand\concept[1]{\textit{#1}}

\newcommand\ccode[1]{\texttt{#1}}

\newcommand{\overarrowi}[1]{\xrightarrow{#1}}
\newcommand{\overarrow}[1]{
  \mathchoice{\raisebox{-3pt}{ $\overarrowi{#1}$ }}
             {\raisebox{-3pt}{ $\overarrowi{#1}$ }}
             {\raisebox{-3pt}{ $\overarrowi{#1}$ }}
             {\raisebox{-3pt}{ $\overarrowi{#1}$ }}}

\providecommand{\tuple}[1]{\ensuremath{\left< #1 \right>}}
\providecommand{\set}[1]{\ensuremath{\left\lbrace #1 \right\rbrace}}

\providecommand{\sizeof}[1]{\ensuremath{\left\vert{#1}\right\vert}}

\defmath{\becomes}{\coloneqq}
\defmath{\true}{{\mathrm{true}}\xspace}
\defmath{\false}{{\mathrm{false}}\xspace}
\defmath{\tbool}{{\mathsf{Bool}}}
\defmath{\sortbool}{{\mathbb{B}}}

\defmath{\nat}{{\mathbb N}}
\defmath{\bool}{{\mathbb B}}
\defmath{\sortnat}{{\mathbb{N}}}

\renewcommand\to{\longrightarrow}

\providecommand{\Assign}[2]{{#1{~:=~}#2}}
\providecommand{\AssignState}[2]{\State \Assign{#1}{#2}}

\newcommand{\defn}{\,\triangleq\,}

\newcommand{\rrestr}[0]{\hspace{-.3mm}\mathrel{\reflectbox{\rotatebox[origin=tl]{-25}{$\|$}}}\hspace{-.3mm}}
\defmath{\lrestr}{\hspace{-.3mm}\mathrel{\reflectbox{\rotatebox[origin=tr]{25}{$\|$}}}\hspace{-.3mm}}

\defmath\Ta{\rightarrow}
\defmath\lmv{M^\shortleftarrow}
\defmath\rmv{M^\shortrightarrow}
\defmath\RR{\mathsf{Pre}}
\defmath\LL{\mathsf{Post}}
\defmath\NN{\mathsf{Ext}}
\defmath\CC{\mathcal{C}}
\defmath\II{\mathcal{I}}
\defmath\nRR{\overline{\RR}}
\defmath\nLL{\overline{\LL}}
\defmath\nNN{\overline{\NN}}
\defmath\nWW{\overline{\WW}}
\defmath\nEE{\overline{\EX}}
\defmath\trtrans{\hookrightarrow}
\defmath\brtrans{\leadsto}
\defmath\wrtrans{\rightsquigarrow}

\defmath\phases{H}
\defmath\tstates{S'}
\defmath\tsint{\sigma_0'}
\defmath\ttrans{T'}
\defmath\vcs{VCS}
\defmath\OO{\textsf{Poxt}}

\defmath\cts{\textsc{ts}}
\defmath\brcts{\mathrel{\hspace{.8mm}\ooalign{%
  \raisebox{1.3\height}{\hspace{-.8mm}$\brtrans$}\cr\hidewidth\cts\hidewidth\cr}}}
\defmath\procs{P}
\defmath\states{S}
\defmath\sint{\sigma_0}
\defmath\trans{T}
\defmath\reach{{\mathcal R}}
\def\red{\widetilde}
\defmath\actions{A}
\defmath\guardof{{G}}
\defmath\guards{\mathcal G}
\newcommand\tr[2]{\ensuremath{\stackrel{#1\phantom{x}}{\longrightarrow_{#2}}}}

\defmath\nes{\mathit{nes}}
\defmath\comm{\stackrel\leftrightarrow\bowtie}
\defmath\lcomm{\stackrel\leftarrow\bowtie}
\defmath\rcomm{\stackrel\rightarrow\bowtie}

\defmath\por{\mathit{por}}
\defmath\en{\mathit{en}}
\defmath\dis{\overline{\mathit{en}}}
\defmath\st{\mathit{st}}
\defmath\sst{\mathit{sst}}
\defmath\fst{\mathit{fst}}

\defmath\portrans{\xdashrightarrow{}}

\title{Stubborn Transaction Reduction (with Proofs)}
\begin{document}

\author{Alfons Laarman}
%\orcid{orcid.org/0000-0002-2433-4174}
%\affiliation{%
%  \institution{TU Wien}
%  \streetaddress{Favorietenstrasse 9/11}
%  \city{Vienna} 
%  \state{Austria} 
%  \postcode{1040}
%}
\institute{Leiden University, Leiden, The Netherlands, \email{a.w.laarman@liacs.leidenuniv.nl}\footnote{%\label{ack}
\scriptsize This work is partially
supported by the Austrian National Research Network S11403-N23 (RiSE)
of the Austrian Science Fund (FWF) and by the Vienna Science and
Technology Fund (WWTF) through grant VRG11-005.
}}
%\email{alfons@laarman.com}


\maketitle

\begin{abstract}
%!TEX root = main.tex
\smaller
The exponential explosion of parallel interleavings
remains a fundamental challenge to model checking of concurrent programs.
Both partial-order reduction (POR) and transaction reduction (TR) decrease the number of interleavings
in a concurrent system. Unlike POR, transactions also reduce the number of intermediate states.
Modern POR techniques, on the other hand, offer more dynamic ways of identifying commutative behavior,
a crucial task for obtaining good reductions.

We show that transaction reduction can use the same dynamic commutativity as found in
stubborn set POR. We also compare reductions obtained by POR and TR, demonstrating with
several examples that these techniques complement each other.

With an implementation of the dynamic transactions in the model checker LTSmin,
we compare its effectiveness with the original static TR
and two POR approaches.
Several inputs, including realistic case studies, demonstrate
that the new dynamic TR can surpass POR in practice.

%TODO: process based?

%We conclude that the open problem of integrating the two methods can lead
%to more powerful model checking approaches for concurrent programs.

 %., mCRL2 and ProB.

\end{abstract}

%\ccsdesc[500]{Computer systems organization~Embedded systems}

%\keywords{model checking, transactions, reduction, Lipton, partial order reduction}

%!TEX root = main.tex

\section{Introduction}
\label{sec:introduction}

POR~\cite{katz-peled,parle89,godefroid} yields state space reductions by selecting a subset $P_\sigma$
of the enabled actions $E_\sigma$ at each state $\sigma$;
the other enabled actions $E_\sigma \setminus P_\sigma$ are pruned.
%The stubborn set method~\cite{parle89} obtains
For instance, reductions preserving deadlocks
(states without outgoing transitions) can be obtained
by ensuring the following properties for the set
$P_\sigma \subseteq E_\sigma \subseteq A$, where $A$ is the set of all actions:

\vspace{-1em}
\begin{wrapfigure}{r}{3.7cm}\vspace{-1.5em}
\begin{tikzpicture}[baseline={([yshift=-.5ex]current bounding box.center)},node distance=1cm]
  \node (s0) {\small $\sigma$};
  \node (s1) [right of=s0,gray] {\small $\sigma_1$};
  \node (s2) [right of=s1,gray] {\small $\sigma_{n-1}$};
  \node (s3) [right of=s2,gray] {\small $\sigma_n$};
  \path (s0) -- node[midway,gray]{\small$\tr{\beta_1}{}$} (s1)
             -- node[midway,gray]{\dots} (s2)
             -- node[midway,gray]{\small$\tr{\beta_n}{}$} (s3);
  
  \node (s0p) [below of=s0] {\small $\sigma'$};
  \node (s1p) [right of=s0p] {\small $\sigma_1'$};
  \node (s2p) [right of=s1p] {\small $\sigma_{n-1}'$};
  \node (s3p) [right of=s2p] {\small $\sigma_n'$};
  \path (s0p) -- node[midway]{\small$\tr{\beta_1}{}$} (s1p)
             -- node[midway]{\dots} (s2p)
             -- node[midway]{\small$\tr{\beta_n}{}$} (s3p);
             
  \path (s0) -- node[midway,sloped]{\small$\tr{\alpha}{}$} (s0p);
  \path (s1) -- node[midway,sloped,gray]{\small$\tr{\alpha}{}$} (s1p);
  \path (s2) -- node[midway,sloped,gray]{\small$\tr{\alpha}{}$} (s2p);
  \path (s3) -- node[midway,sloped,gray]{\small$\tr{\alpha}{}$} (s3p);
\end{tikzpicture}
\vspace{-2.5em}
\end{wrapfigure}~
\begin{itemize}
\item
In any state $\sigma_n$ reachable from $\sigma$ via pruned actions
$\beta_1,\dots,\beta_n \in A\setminus P_\sigma$, all actions
  $\alpha \in P_\sigma$ \concept{commute} with the pruned actions $\beta_1,\dots,\beta_n$
  %those in $E_\sigma \setminus P_\sigma$, 
  and
\item at least one action $\alpha\in P_\sigma$ remains enabled in $\sigma_n$.
\end{itemize}
The first property ensures that the pruned actions $\beta_1,\dots,\beta_n$
are still enabled after $\alpha$ and lead to the same state ($\sigma_n'$),
i.e., the order of executing $\beta_1,\dots,\beta_n$ and $\alpha$ is irrelevant.
The second avoids that deadlocks are missed when
pruning states $\sigma_1,\ldots,\sigma_n$.
To compute the POR set $P_\sigma$ without computing pruned states $\sigma_1,\ldots,\sigma_n$
(which would defeat the purpose of the reduction it is trying to attain in the first place),
\emph{Stubborn POR uses static analysis to `predict' the future from $\sigma$, i.e.,
to over-estimate the $\sigma$-reachable actions $A\setminus P_\sigma$, e.g.: $\beta_1,..,\beta_n$.}
% in \actions.
%more static variants of the stubborn set overestimated reachability behavior
%by reasoning over the disabled actions and their
%relative enabling, in order to exclude possible future non-commutativity.

\begin{figure}[b]\vspace{-.5em}
\centering
\scalebox{.66}{
%!TEX root = ../main.tex

\begin{tikzpicture}[node distance=1.7cm,font=\large]
\newcommand\xstate[3]{(\tuple{#1,#2},#3)}
\renewcommand\xstate[3]{\tuple{#3}}

\tikzset{snake/.style={
decoration={snake, 
    amplitude = .4mm,
    segment length = 2mm,
    post length=0.9mm},decorate}}

\def\TONE{}
\def\TTWO{}

%   \tikzstyle{every node}=[minimum width=1cm]
\tikzstyle{e}=[line width=2pt]

    \node (s0) 	 					{};%\parbox{.5cm}{$\xstate{1}{a}{x,y}$\\$\xstate{1}{a}{0,0}$}};
    \node (s0p) 	 					{};%\parbox{.5cm}{$\xstate{1}{a}{x,y}$\\$\xstate{1}{a}{0,0}$}};
    	%\parbox{.5cm}{$\xstate{1}{a}{x,y}$\\$\xstate{1}{a}{0,0}$}};
    \node (s1) [below left	of=s0p] 	{};%$\xstate{2}{a}{0,0}$};
    \node (s2) [below right of=s0p] 	{};%$\xstate{1}{b}{0,1}$};
    \node (s3) [below left 	of=s1] 	{};%$\xstate{3}{a}{0,2}$};
    \node (s4) [below right of=s1]	{};%$\xstate{2}{b}{0,1}$};
    \node (s5) [below right of=s2] 	{};%$\xstate{1}{c}{1,1}$};
    \node (s6) [below right	of=s3] 	{};%$\xstate{3}{b}{0,2}$};
    \node (s7) [below right	of=s4] 	{};%$\xstate{2}{c}{1,1}$};
    \node (s8) [below left	of=s7] 	{};%$\xstate{2}{c}{1,1}$};

    \path (s0p) edge[e,->] node(t1')[midway,above,sloped] {\ccode{a=0;}}
    					node(t1')[midway,below,sloped] {$\TONE$}  (s1);
    \path (s0p) edge[->] node(t1')[midway,above,sloped] {\ccode{x=1;}}
    					node(t1')[midway,below,sloped] {$\TTWO$} (s2);
    \path (s1) edge[->] node(t1')[midway,above,sloped] {\ccode{b=2;}}
    					node(t1')[midway,below,sloped] {$\TONE$} (s3);
    \path (s1) edge[e,->] node(t1')[midway,above,sloped] {\ccode{x=1;}}
					    node(t1')[midway,below,sloped] {$\TTWO$} (s4);
    \path (s2) edge[->] node(t1')[midway,above,sloped] {\ccode{a=0;}}
					    node(t1')[midway,below,sloped] {$\TONE$} (s4);
    \path (s2) edge[->] node(t1')[midway,above,sloped] {\ccode{y=2;}}
					    node(t1')[midway,below,sloped] {$\TTWO$} (s5);
    \path (s4) edge[e,->] node(t1')[midway,above,sloped] {\ccode{b=2;}}
					    node(t1')[midway,below,sloped] {$\TONE$} (s6);
    \path (s4) edge[->] node(t1')[midway,above,sloped] {\ccode{y=2;}}
    					node(t1')[midway,below,sloped] {$\TTWO$} (s7);
    \path (s3) edge[->] node(t1')[midway,above,sloped] {\ccode{x=1;}}
					    node(t1')[midway,below,sloped] {$\TTWO$} (s6);
    \path (s5) edge[->] node(t1')[midway,above,sloped] {\ccode{a=0;}}
    					node(t1')[midway,below,sloped] {$\TONE$} (s7);
    \path (s7) edge[->] node(t1')[midway,above,sloped] {\ccode{b=2;}}
					    node(t1')[midway,below,sloped] {$\TONE$} (s8);
    \path (s6) edge[e,->] node(t1')[midway,above,sloped] {\ccode{y=2;}}
    					node(t1')[midway,below,sloped] {$\TTWO$} (s8);

    \path (s0p) edge[->,e,snake,bend right=45] node(n1)[midway,above,sloped] {\ccode{a=0;b=2;}} (s3);
    \path (s0p) edge[->,e,snake,bend left=45] node(n2)[midway,above,sloped] {\ccode{x=1;y=2;}} (s5);
    \path (s5) edge[->,e,snake,bend left=45] node(t1')[midway,below,sloped] {\ccode{a=0;b=2;}} (s8);
    \path (s3) edge[->,e,snake,bend right=45] node(t1')[midway,below,sloped] {\ccode{x=1;y=2;}} (s8);

    \node (s0x) 	[node distance=8.5cm,right of=s0p] {};
    \node (s0y) 	[node distance=.7cm,above of=s0x] {};
      \node (s1) [below left	of=s0x] 	{};%$\xstate{2}{a}{0,0}$};
    \node (s2) [below right of=s0x] 	{};%$\xstate{1}{b}{0,1}$};
    \node (s3) [below left 	of=s1] 	{};%$\xstate{3}{a}{0,2}$};
    \node (s4) [below right of=s1]	{};%$\xstate{2}{b}{0,1}$};
    \node (s5) [below right of=s2] 	{};%$\xstate{1}{c}{1,1}$};
    \node (s6) [below right	of=s3] 	{};%$\xstate{3}{b}{0,2}$};
    \node (s7) [below right	of=s4] 	{};%$\xstate{2}{c}{1,1}$};
    \node (s8) [below left	of=s7] 	{};%$\xstate{2}{c}{1,1}$};

    \path (s0x) edge[e,->] node(t1')[midway,above,sloped] {\ccode{a=0;}}
    					node(t1')[midway,below,sloped] {$\TONE$}  (s1);
    \path (s0x) edge[->] node(t1')[midway,above,sloped] {\ccode{x=1;}}
    					node(t1')[midway,below,sloped] {$\TTWO$} (s2);
    \path (s1) edge[->] node(t1')[midway,above,sloped] {\ccode{b=2;}}
    					node(t1')[midway,below,sloped] {$\TONE$} (s3);
    \path (s1) edge[e,->] node(t1')[midway,above,sloped] {\ccode{x=1;}}
					    node(t1')[midway,below,sloped] {$\TTWO$} (s4);
    \path (s2) edge[->] node(t1')[midway,above,sloped] {\ccode{a=0;}}
					    node(t1')[midway,below,sloped] {$\TONE$} (s4);
    \path (s2) edge[->] node(t1')[midway,above,sloped] {\ccode{y=2;}}
					    node(t1')[midway,below,sloped] {$\TTWO$} (s5);
    \path (s4) edge[e,->] node(t1')[midway,above,sloped] {\ccode{b=2;}}
					    node(t1')[midway,below,sloped] {$\TONE$} (s6);
    \path (s4) edge[->] node(t1')[midway,above,sloped] {\ccode{y=2;}}
    					node(t1')[midway,below,sloped] {$\TTWO$} (s7);
    \path (s3) edge[->] node(t1')[midway,above,sloped] {\ccode{x=1;}}
					    node(t1')[midway,below,sloped] {$\TTWO$} (s6);
    \path (s5) edge[->] node(t1')[midway,above,sloped] {\ccode{a=0;}}
    					node(t1')[midway,below,sloped] {$\TONE$} (s7);
    \path (s7) edge[->] node(t1')[midway,above,sloped] {\ccode{b=2;}}
					    node(t1')[midway,below,sloped] {$\TONE$} (s8);
    \path (s6) edge[e,->] node(t1')[midway,above,sloped] {\ccode{y=2;}}
    					node(t1')[midway,below,sloped] {$\TTWO$} (s8);

%    \path (s0x) edge[->,e,snake,bend right=45] node(n1)[midway,above,sloped] {\ccode{a=0;b=2;}} (s3);
%    \path (s0x) edge[->,e,snake,bend left=45] node(n2)[midway,above,sloped] {\ccode{x=1;y=2;}} (s5);
%    \path (s5) edge[->,e,snake,bend left=45] node(t1')[midway,below,sloped] {\ccode{a=0;b=2;}} (s8);
%    \path (s3) edge[->,e,snake,bend right=45] node(t1')[midway,below,sloped] {\ccode{x=1;y=2;}} (s8);

    \path (s0y) edge[e,->] node(t1')[midway,right] {\ccode{a=b=x=y=0;}} (s0x);
    
\end{tikzpicture}
}
%\scalebox{.8}{
%\input{figures/excomp2}
%}
\caption{Transition systems of $\mathit{program1}$ (left) and $\mathit{program2}$ (right).
%Transitions of Thread 1 are all directed in left-down-ward direction,
%transitions of Thread 2 right-down-ward, and
%transitions of the added main thread vertically.
Thick lines show optimal (Stubborn set) POR.
Curly lines show a TR (not drawn in the right figure). % (Ample set irreducable.)
%Static TR yields no reduction for the 
%Both (stubborn-set) POR and TR yield no reductions (thick arrows) for the system
%with dependencies . The Stubborn TR presented in the current paper does yield reductions
%for this system (shown in curly braces).
%(c.f. \cite{peled-93,parle89,godefroid}).
}
\label{f:lipton}\vspace{-.5em}
\end{figure}
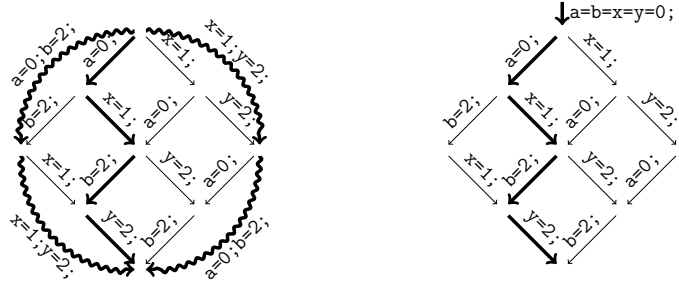

Lipton or transaction reduction (TR)~\cite{lipton}, on the other hand, identifies
sequential blocks in the actions $\actions_i$ of each thread $i$
that can be grouped into transactions.
A transaction $\alpha_1..\alpha_k..\alpha_n \in \actions_i^*$
is replaced with an atomic action $\alpha$ which is its sequential composition, i.e.
$\alpha = \alpha_1\circ..\circ\alpha_k\circ..\circ\alpha_n$.
Consequently, any trace
$\sigma_1 \tr{\alpha_1}{} \sigma_2 \tr{\alpha_2}{}
\ldots\tr{\alpha_k}{}\ldots
\tr{\alpha_{n-1}}{} \sigma_n \tr{\alpha_n}{}\sigma_{n+1}$
is replaced by $\sigma_1 \tr{\alpha}{} \sigma_{n+1}$,
making state $\sigma_2, \ldots, \sigma_n$ \concept{internal}.
Thereby, internal states disallow all interleavings of other threads $j\neq i$,
i.e., \concept{remote actions} $\actions_j$ are not fired at these states.
Similar to POR, this pruning can reduce reachable states.
Additionally, internal states
can also be discarded when irrelevant to the model checking problem.

%Transactions were developed in the context of concurrent databases.
In the database terminology of origin~\cite{papadimitriou}, a transaction must consist of:
\begin{itemize}
\item A \concept{pre-phase},
		containing actions $\alpha_1..\alpha_{k-1}$ that may gather required resources,
\item a single \emph{commit action} $\alpha_k$
		possibly interfering with remote actions, and
\item a \concept{post-phase} $\alpha_{k+1}..\alpha_{n}$,
		possibly releasing resources (e.g. via unlocking them).
\end{itemize}
In the pre- and post-phase, the actions (of a thread $i$)
must commute with all remote behavior, i.e.
\emph{all actions $\actions_j$ of all other threads $j\neq i$ in the system}.

\emph{TR does not dynamically `predict' the possible future remote actions,
like POR does.}
This makes the commutativity requirement
needlessly stringent, as the following example shows:
Consider $\mathit{program1}$ consisting of two threads. All actions of one thread commute with
all actions of the other because only local variables are accessed.
\autoref{f:lipton} (left) shows the POR and TR of this system.
\[\footnotesize
\mathit{program1} := \ccode{\textbf{if} (fork()) \{a = 0; b = 2; \} \textbf{else} \{ x = 1; y = 2; \}}
\]\vspace{-2em}
\[\footnotesize
\mathit{program2} := \ccode{a = b = x = y = 0; \textbf{if} (fork()) \{ $\mathit{program1}$; \}}
\]
Now assume that a parallel assignment is added as initialization code
yielding $\mathit{program2}$ above.
\autoref{f:lipton} (right) shows again the reductions.
Suddenly, all actions of both threads become dependent on the initialization,
i.e. neither action \ccode{a = 0;} nor action \ccode{b = 2;} commute with
actions of other threads, spoiling the formation of a transaction
\ccode{\textsf{atomic}\{a = 0; b = 2;\}} (idem for \ccode{\textsf{atomic}\{x = 1; y = 2;\}}).
Therefore, TR does not yield any reduction anymore (not drawn).
Stubborn set POR~\cite{valmari1988error}, however, still reduces \textit{program2} like \textit{program1},
because, using static analysis, it `sees' that the initialization cannot be fired again.\footnote
{\scriptsize\textit{program2} is a simple example. Yet various programming patterns
lead to similar behavior, e.g.: lazy initialization,
atomic data structure updates and load balancing~\cite{vmcai}.}

In the current paper, we show how TR can be made dynamic in the same sense as
stubborn set POR~\cite{parle89},
so that the previous example again yields the maximal reduction.
Our work is based on the prequel~\cite{vmcai},
where we instrument programs 
%with dynamic commutativity conditions 
in order to obtain dynamic TR for \emph{symbolic model checking}.
While \cite{vmcai} premiered dynamically growing and shrinking transactions,
%beyond only locked regions as in \cite{qadeer-transactions} (see \autoref{sec:related}),
its focus on symbolic model checking complicates a direct
comparison with other dynamic techniques such as POR.
The current paper therefore extends this technique to enumerative model checking,
which allows us to get rid of the heuristic conditions from~\cite{vmcai}
by replacing them with the more general stubborn set POR method.
%That instrumentation avoids
%having to encode the complicated POR rules in the transition
%relation, a typical bottleneck when combining POR and \emph{symbolic}
%approaches~\cite{burch-clarke-mcmillan-dill-hwang}.
While we can reduce the results in the current paper to the 
reduction theorem of~\cite{vmcai}, the new focus on enumerative model checking
provides opportunities to tailor reductions on a
per-state basis and investigate TR more thoroughly.\footnote{\label{fn:symbolic}Symbolic approaches can be viewed as reasoning over sets of states, and
therefore cannot easily support fine-grained per-state POR/TR analyses.}
This leads to various contributions:

\begin{enumerate}
%\item An adaptation of stubborn set POR for TR more dynamic than TR in~\cite{vmcai}.
	%. The resulting
	 %`stubborn TR' is more dynamic and general than~\cite{vmcai} as it lacks heuristic conditions.
\item A `Stubborn' TR algorithm (STR)
	more dynamic/general than TR in~\cite{vmcai}.
\item An open source implementation of (stubborn) TR in the model~checker~\textsc{LTSmin}.
\item Experiments comparing TR and POR for the first time in \autoref{sec:experiments}.
%\item A comparison with reductions from the \textsc{SPIN} model checker on \textsc{Promela} models.
\end{enumerate}

%In particular, we show that POR and TR offer incomparable reductions
%(neither is strictly stronger than the other).
%In addition, stubborn transaction reduction offers several advantages:

Moreover, in \autoref{sec:comparison}, we show analytically that unlike stubborn POR:
\begin{enumerate}
\item Computing optimal stubborn TR is tractable and reduction is not heuristic.
\item Stubborn TR can exploit right-commutativity and prune (irrelevant) deadlocks
(while still preserving invariants as per \autoref{th:alg}).
%\item Influence of visible transitions seems more limited.
%\item Stubborn TR can prune deadlocks.
\end{enumerate}
On the other hand, stubborn POR is still more effective for
checking for absence of deadlocks and reducing massively parallel systems.
Various open problems, including the combination of TR and POR,
leave room for improvement.

The current paper is the technical report version of \cite{laarman18}.
Proofs of theorems and lemmas can be found in \autoref{app:proofs}.

%conclude from both our theoretical and practical analyses
%that a combination of POR and TR could offer further improvements
%for model checking of concurrent systems.
%We identify this as an open problem and point out
%various other open problems.

%In \autoref{sec:related}, we discuss related work and
% and future work, including
%\emph{dynamic partial order reduction}. Indeed, the above-discussed POR methods
%are often ``static'' because these methods require that commutativity
%derived a-priory over statically grouped actions. This commutativity relation
%is in practice over-estimated by running static analysis in advance.
%Nonetheless, we may observe that the above-mentioned POR approach is
%still \concept{semi-dynamic} because it calculates reductions on a per-state basis:
%If a dependent action no longer can be fired from the current state, then
%more reduction becomes possible.
%
%Using a simple observation discussed in the section, we point out that
%semi-dynamic approaches can be made completely dynamic.
%We merely need to derive a commutativity condition for each action that
%evaluates to true in states where the action commutes. Using this condition,
%all actions can be split in a commuting and a non-commuting part, the first
%guarded by the condition and the second by the negated condition.
%The 
%This technique is similar to how we encoded dynamic reductions for
%symbolic model checking. We argue that it can also be used to make
%stubborn sets fully dynamic.
%In \autoref{sec:conclusion}, we

%!TEX root = main.tex

\vspace{-.5em}
\section{Preliminaries}
\label{sec:prelim}

\vspace{-.5em}

\paragraph{Concurrent transition systems}
We assume a general process-based semantic model that can
accommodate various languages.
A concurrent transition system (CTS) for a finite set of processes \procs is tuple
$\cts \defn \tuple{\states, \trans, \actions, \sigma_0}$
with finitely many actions $\actions \defn \biguplus_{i\in P} \actions_i$.
Transitions are relations between states and actions: $\trans\subseteq	\states\times\actions\times \states$.
We write $\alpha_i$ for $\alpha \in \actions_i$,
$\sigma \tr{\alpha}i \sigma'$ for $\tuple{\sigma,\alpha_i,\sigma'}\in T$,
$\trans_i$ for $\trans \cap (\states\times\actions_i\times \states)$,
$\trans_\alpha$ for $\trans \cap (\states\times\set{\alpha}\times \states)$,
$\tr{\alpha}{}$ for $\set{\tuple{\sigma,\sigma'}\mid \tuple{\sigma,\alpha,\sigma'}\in T}$, and
$\tr{}{i}$ for $\set{\tuple{\sigma,\sigma'}\mid \tuple{\sigma,\alpha,\sigma'}\in T_i}$.
%%A guard $g\in\guards$ is a predicate over states and thus
%%constitutes subsets of states: $g \subseteq \states$.
%Actions can have multiple conjunctive guards:~$\guardof \colon \actions \rightarrow \powerset{\guards}$,
%where \guards is the guard set.
%We write $g(\sigma)$ if $g$ evaluates to true in state $\sigma$.
%A sound CTS has:
%$\tuple{\sigma,\alpha,\sigma'}\in T \Longleftrightarrow \forall {g\in G(\alpha)}\colon g(\sigma)$.
%A guard is enabled by action $\alpha$ if there is $\sigma\tr{\alpha}{}\sigma'$ with
%$\neg g(\sigma), g(\sigma')$.
%POR reasons over enabling of actions using their conjunctive guards.
%%or abbreviated:
%%$\tuple{\sigma,\alpha,\sigma'}\in T \Leftrightarrow G(\alpha)(s)$.

State space exploration can be used to show invariance of a property $\varphi$, e.g.,
expressing mutual exclusion, written: $\reach(\cts) \models \varphi$.
This is done by finding all reachable states $\sigma$, i.e.,
$\reach(\cts)\defn \set{\sigma \mid \sigma_0 \rightarrow^* \sigma}$,
and show that $\sigma\in \varphi$.

%The goal of state space exploration is to find the reachable states of the CTS:
%$\reach(\cts)\defn \set{\sigma \mid \sigma_0 \rightarrow^* \sigma}$
%(and thereby show absence of errors in the system).

\paragraph{Notation}
We let $\en(\sigma)$ be the set of actions enabled at $\sigma$:
\set{\alpha\mid \exists \tuple{\sigma,\alpha,\sigma'}\in T } and
$\overline{\en}(\sigma) \defn \actions \setminus \en(\sigma)$.
We let $R\circ Q$ and $RQ$ denote the \concept{sequential composition}
of two binary relations $R$ and $Q$, defined~as:
$
\{ (x,z)\,\vert\,\exists y\colon (x,y)\in R
\wedge (y,z)\in Q \}\,
$.
Let $R \subseteq S \times S$ and $X \subseteq S$. Then
left restriction of $R$ to $X$ is $X \lrestr R \defn R\,\cap\,(X \times S)$ and right restriction
is $R\rrestr X \defn R\,\cap\,(S \times X)$.
%$\overline X \defn S  \setminus X$
%The relations $X \lrestr R$ and $R \rrestr X$
%are the left/right restriction~of~$R$~to~$X$. 
The complement of $X$ is denoted
$\overline X \defn S  \setminus X$ (the universe of all states remains implicit in this notation).
The inverse of $R$ is $R^{-1}\defn \set{\tuple{x,y}\mid \tuple{y,x}\in R}$.
%The (implicit) universe $U$ is always taken to be the set of all possible states ($S$).
%Finally, $R$ does not enable a predicate (or set of states) $X$ iff $\overline X\lrestr R \rrestr  X = \emptyset$.
%Not disabling $X$ is not enabling $\overline{X}$.
%and $R$ does not disable $X$ iff $X\lrestr R \rrestr  \overline X = \emptyset$.

\paragraph{POR relations}
Dependence is a well-known relation used in POR. 
Two actions $\alpha_1,\alpha_2$ are dependent if there is a state where they do not commute,
hence we first define commutativity.
Let $c \defn \set{{\sigma}\mid\exists  \tuple{\sigma,\alpha_1,\sigma'},\tuple{\sigma,\alpha_2,\sigma''}\in \trans}$. Now:
\begin{IEEEeqnarray}{rCrll}
\tr{\alpha_1}{} \comm \tr{\alpha_2}{}%, \alpha_1 \comm \alpha_2
		&\defn & \,\,c\lrestr \tr{\alpha_1}{} \circ \tr{\alpha_2}{}
							&~=~ c \lrestr \tr{\alpha_2}{} \circ \tr{\alpha_1}{}
							&\text{~($\alpha_1$, $\alpha_2$
							strongly-commute)} \nonumber\\*
\tr{\alpha_1}{} \bowtie \tr{\alpha_2}{}%, \alpha_1\bowtie \alpha_2
		&\defn & \tr{\alpha_1}{} \circ \tr{\alpha_2}{}
						 	&~=~ \tr{\alpha_2}{} \circ \tr{\alpha_1}{}
							&\text{~($\alpha_1$, $\alpha_2$ 
							commute, also $\alpha_1 \bowtie \alpha_2$)} \nonumber\\*
\tr{\alpha_1}{} \rcomm \tr{\alpha_2}{}%, \alpha_1\rcomm \alpha_2
		&\defn &\tr{\alpha_1}{} \circ \tr{\alpha_2}{}
							&~\subseteq~ \tr{\alpha_2}{} \circ \tr{\alpha_1}{}
							&\text{~($\alpha_1$ 
							right-commutes with $\alpha_2$)}\nonumber\\*
\tr{\alpha_1}{} \lcomm \tr{\alpha_2}{}%, \alpha_1\lcomm \alpha_2
		&\defn &\tr{\alpha_1}{} \circ \tr{\alpha_2}{}
							&~\supseteq~ \tr{\alpha_2}{} \circ \tr{\alpha_1}{}
							&\text{~($\alpha_1$  
							left-commutes with $\alpha_2$)}\nonumber
\vspace{-3ex}
\end{IEEEeqnarray}
\begin{minipage}{.485\linewidth}
\begin{align}
\text{%!TEX root = ../main.tex

\hspace{-.6cm}
\begin{tikzpicture}[baseline={([yshift=-.5ex]current bounding box.center)}]

  \node (s0) {\small $\sigma_1$};

%  \node (start) [left of=s0, yshift=-.45cm,xshift=-1.5cm] {\small $\Comm(\rightarrow_t, \rightarrow_u)$ iff};
  
  \node (s1) [node distance=1cm,below of=s0] {\small $\sigma_2$};
  \node (s2) [node distance=.1cm,right of=s1, xshift=.8cm] {\small $\sigma_3$};
  \path (s0) -- node(m)[midway,sloped]{\small$\tr{\alpha_1}{}$} (s1);
  \path (s1) -- node[midway]{\small$\tr{\alpha_2}{}$} (s2);

  \node [right of=s2,yshift=.55cm,xshift=-.5cm]{\small$\Rightarrow\exists \sigma_4:$};

  \node [left of=m,xshift=-.1cm]
   {\small $\forall \sigma_1,\sigma_2,\sigma_3:$};

  \node (s3) [node distance=.3cm,right of=s0, xshift=1.9cm]{\small $\sigma_1$};
  \node (s4) [node distance=.3cm,right of=s3, xshift=.6cm] {\small $\sigma_4$};
  \node (s5) [node distance=1cm,below of=s4] {\small $\sigma_3$};
  \path (s4) -- node[midway,sloped]{\small$\tr{\alpha_1}{}$} (s5);
  \path (s3) -- node[pos=.45]{\small$\tr{\alpha_2}{}$} (s4);

  % grey stuff
  \node (s6) [gray,node distance=1cm,below of=s3] {\small $\sigma_2$};
  \path (s3) -- node(m)[gray,midway,sloped]{\small$\tr{\alpha_1}{}$} (s6);
  \path (s6) -- node[gray,pos=.48]{\small$\tr{\alpha_2}{}$} (s5);

\end{tikzpicture}
\hspace{-.6cm}}\label{eq:right}
\end{align}
\end{minipage}
\begin{minipage}{.515\linewidth}
\begin{align}
\text{%!TEX root = ../main.tex

\hspace{-.6cm}
\begin{tikzpicture}[baseline={([yshift=-.5ex]current bounding box.center)}]

  \node (s0)  {\small $\sigma_1$};

%  \node (start) [left of=s0, yshift=-.45cm,xshift=-1.5cm] {\small $\Comm(\rightarrow_t, \rightarrow_u)$ iff};
  
  \node (s1) [node distance=1cm,below of=s0] {\small $\sigma_2$};
  \node (s2) [node distance=.1cm,right of=s0, xshift=.8cm] {\small $\sigma_3$};
  \path (s0) -- node(m)[midway,sloped]{\small$\tr{\alpha_1}{}$} (s1);
  \path (s0) -- node[midway]{\small$\tr{\alpha_2}{}$} (s2);

  \node [left of=m,xshift=.4cm]
   {\hspace{-2em}\small $\forall \sigma_1,\sigma_2,\sigma_3\colon$};
  \node [right of=s2,yshift=-.55cm,xshift=-.3cm]{\small$\Rightarrow \exists \sigma_4:$};

  % grey stuff
  \node (s3) [gray,node distance=.6cm,right of=s0, xshift=1.7cm]{\small $\sigma_1$};
  \path (s3) -- node[gray,pos=.45]{\small$\tr{\alpha_2}{}$} (s4);
  \path (s3) -- node(m)[gray,midway,sloped]{\small$\tr{\alpha_1}{}$} (s6);

  \node (s4) [node distance=.3cm,right of=s3, xshift=.6cm] {\small $\sigma_3$};
  \node (s6) [node distance=1cm,below of=s3] {\small $\sigma_2$};
  \node (s5) [node distance=1cm,below of=s4] {\small $\sigma_4$};
  \path (s4) -- node[midway,sloped]{\small$\tr{\alpha_1}{}$} (s5);
  \path (s6) -- node[pos=.48]{\small$\tr{\alpha_2}{}$} (s5);

\end{tikzpicture}
\hspace{-.8cm}}\label{eq:strong}
\end{align}
\end{minipage}
Left / right commutativity allows actions to be
prioritized / delayed over other actions without affecting the end state.
\autoref{eq:right} illustrates this by quantifying of the states:
Action $\alpha_1$ right-commutes with $\alpha_2$, and vice verse
$\alpha_2$ left-commutes with $\alpha_1$.
Full commutativity ($\bowtie$) always allows both delay and prioritization for
any serial execution of $\alpha_1, \alpha_2$, while
strong commutativity only demands full commutativity when both actions are
simultaneously enabled, as shown in \autoref{eq:strong}
for deterministic actions $\alpha_1$/$\alpha_2$
(\autoref{eq:strong} is only for an intuition and does not illustrate
the non-deterministic case, which is covered by $\comm$).
%where $A,B,C$ are non-empty sets of states, abusing notation to cover the case with non-deterministic actions $\alpha_1$/$\alpha_2$.
%(unlike \comm, this illustration holds only for deterministic $\alpha_1$ and $\alpha_2$).
Left / right / strong \concept{dependence} implies lack of left / right / strong commutativity, e.g.:
$\alpha_1 \not\bowtie \alpha_2$.
%Slightly abusing notation, we abbreviate sets of
%dependent transitions as $\not\bowtie(\alpha_1)$.
%$\tuple{\alpha_1,\alpha_2}\in \deps \implies \neg  \alpha_1\comm\alpha_2$

Note that typically: $\forall i, \alpha,\beta \in \actions_i \colon \alpha \not\bowtie \beta$
due to e.g. a shared program counter.
Also note that if $\alpha_1 \rcomm \alpha_2$, then $\alpha_1$ never enables $\alpha_2$,
while strong commutativity implies that neither $\alpha$ disables $\beta$, nor vice versa.
%The same holds for 
%Vice versa, if $\alpha_1 \lcomm \alpha_2$, then $\alpha_1$ never disables $\alpha_2$,
%unless $\alpha_2$ (always) disables $\alpha_1$ as well.
%Given that strong commutativity implies right commutativity for both actions,
%it would be redundant to additionally require that they never enable each other.
%Enforcing the absence of mutual disabling, however, strengthens the relation.

A lock(/unlock) operation right(/left)-commutes with other locks and unlocks.
%Similarly, an unlock left-commutes with locks/unlocks.
Indeed, a lock never enables another lock or unlock.
Neither do unlocks ever disable other unlocks or locks.
In the absence of an unlock however, a lock also attains left-commutativity
as it is mutually disabled by other locks.
Because of the same disabling property, two locks however do not strongly commute.

Finally, a \concept{necessary enabling set} (NES) of %a guard $g$ and 
an action $\alpha$ and a state $\sigma_1$
is a set of actions that must be executed for $\alpha$ to become enabled,~formally:\\
%$\nes(\sigma,g) \supseteq \set{\alpha \in A\mid  \exists \tuple{\sigma_1,\alpha,\sigma_2}\in T
%\colon \neg g(\sigma_1) \land g(\sigma_2) \land \sigma\rightarrow^* \sigma_1}$.
$\forall E\in \nes_{\sigma_1}(\alpha), \sigma_1\xrightarrow{\alpha_1,..,\alpha_n} \sigma_2\colon
% \exists \tuple{\sigma_1,\alpha,\sigma_2}\in T
%\colon
\alpha \in \dis(\sigma_1) \land \alpha\in\en(\sigma_2) \implies E \cap \set{\alpha_1,..,\alpha_n}\neq\emptyset$.
An example of an action $\alpha$ with two NESs $E_1,E_2 \in \nes_\sigma(\alpha)$
is a command guarded by $g$ in an imperative language:
When $\alpha\in \dis(\sigma)$, then either its guard $g$ does not hold in $\sigma$,
and $E_1$ consists of all actions enabling $g$, or its program counter is not
activated in $\sigma$, and $E_2$ consists of all actions that label the edges 
immediately before $\alpha$ in the CFG of the process that $\alpha$ is part of.

\paragraph{POR}
POR uses the above relations to find a subset of enabled actions $\por(\sigma) \subseteq \en(\sigma)$
sufficient for preserving the property of interest.
Commutativity is used to ensure that 
the sets $\por(\sigma)$ and $\en(\sigma) \setminus \por(\sigma)$ commute, while
the NES is used to ensure that this mutual commutativity holds \emph{in all future behavior}.
The next section explains how stubborn set POR achieves this.

POR gives rise to a CTS $\red\cts \defn \langle \states, \red\trans, \actions, \sigma_0 \rangle$,
$\red T \defn \set{\tuple{\sigma,\alpha,\sigma'}\in T \mid \alpha\in\por(\sigma)}$,
abbreviated $\sigma \xdashrightarrow{\alpha} \sigma'$.
It is indeed reduced, since we have $\reach(\red\cts) \subseteq \reach(\cts)$.

\paragraph{Transaction reduction}
%For completeness sake, we reiterate a brief explanation of
(Static) transaction reduction % from~\cite{vmcai}.
was devised by Lipton~\cite{lipton}.
%to transform
%``an interruptible routine to an uninterruptible routine that cannot be 
%interleaved with the rest of the system'' in order to simplify correctness
%proofs of parallel programs (``a system of processes'').
%Applied to model checking this method can 
It merges multiple sequential
statements into one atomic operation,
thereby radically reducing the reachable states.
An action $\alpha$ is called %a transition $\stackrel{\alpha}{\to}_i$
a right/left mover if and only if
it commutes with actions from all other threads $j\neq i$:
%abbreviated as $\tr{}{\neq i} \defn \bigcup_{j\neq i} \tr{}j$:
%$\forall j\neq i\colon \alpha_i\circ\rightarrow_j\subseteq\rightarrow_j\circ\alpha_i$.
\vspace{-2mm}\[\tr{\alpha}{i} \,\,\rcomm\,\,  \bigcup_{j\neq i} \tr{}{j}
		\text{~($\alpha$ is a right mover)\phantom{XX}}
	 \tr{\alpha}{i} \,\,\lcomm\,\,  \bigcup_{j\neq i} \tr{}{j}
	 	 ~\text{~($\alpha$ is a left mover)} 
\]\vspace{-4mm}

\noindent
\concept{Both-movers} are transitions that are both left and right movers,
whereas \concept{non-movers} are neither.
The sequential composition of two movers is also a corresponding mover,
and vice versa.
Moreover, one may always safely classify an action as a non-mover, although
having more movers yields better reductions.

Examples of right-movers are locks, P-semaphores and synchronizing
queue operations. Their counterparts; unlock, V-semaphore and enqueue ops, are left-movers.
Their behavior is discussed above using locks and unlocks as an example.

Lipton reduction only preserves halting.
We present Lamport's~\cite{lamport-lipton} version, which preserves
safety properties such as $\Box \varphi$, i.e. $\varphi$ is an invariant.
%Lipton aimed at preserving halting in the reduced system.
%He showed that 
Any sequence $\alpha_1,\ldots, \alpha_n$
%$\tr{\alpha_1}{i}\circ \tr{\alpha_2}{i}\circ
%\dots \circ \tr{\alpha_{n-1}}{i} \circ \tr{\alpha_n}{i}$
%$\alpha_1;\ldots;\alpha_k;\
can be
\concept{reduced} to a single action %\concept{transaction} %or \concept{atomic section}
$\alpha$ s.t. $\tr{\alpha}{i} = \tr{\alpha_1}{i}\circ \ldots \circ\tr{\alpha_n}{i}$
(i.e. a compound statement with the same local behavior),
if for some $1 \le k < n$:
\begin{lipton}[parsep=0pt]%[noitemsep, topsep=0pt, parsep=0pt, partopsep=0pt]
\item\label{L1}
	 actions before the commit $\alpha_k$ are right movers:
	$\tr{\alpha_1}{i}\circ \ldots\circ \tr{\alpha_{k-1}}{i} \,\,\rcomm\,\, \tr{}{\neq i}$,
\item\label{L2}
	 actions after the commit $\alpha_{k}$ are left movers:
	$\tr{\alpha_{k+1}}{i}\circ \ldots\circ \tr{\alpha_{n}}{i}  \,\,\lcomm\,\, \tr{}{\neq i}$,
\item\label{L3}
	 actions after $\alpha_1$ do not block:
	$\forall \sigma\, \exists \sigma' \colon
				\sigma\tr{\alpha_{1}}{i}\circ \ldots\circ \tr{\alpha_{n}}{i}\sigma'$, and
	\vphantom{$\lcomm\bigcup_{j\neq i}$}
\item\label{L4} 
	$\varphi$ is not disabled by $\tr{\alpha_1}{i}\circ \ldots\circ \tr{\alpha_{k-1}}{i}$, nor enabled by
	 $\tr{\alpha_{k+1}}{i}\circ \ldots\circ \tr{\alpha_{n}}{i}$.\parbox{0pt}{\vphantom{$\lcomm\bigcup_{j\neq i}$}}
\end{lipton}

\begin{wrapfigure}{r}{5cm}
\vspace{-2.5em}
\hspace{-1em}
\scalebox{.75}{
\begin{tikzpicture}

   \tikzstyle{e}=[minimum width=0cm]
   \tikzstyle{every node}=[font=\small, node distance=.9cm, inner sep=1pt]

	\node (s1) [e] {$\sigma_1$};
	\node (s2) [right of=s1,e] {$\sigma_2$};
	\node (s3) [right of=s2,e] {$\sigma_3$};
	\node (s4) [right of=s3,e] {$\sigma_4$};
%	\node (s5) [right of=s4,e] {$\sigma_5$};
	\node (s6) [right of=s4,e] {$\sigma_5$};
	\node (sx) [right of=s6,e] {$\sigma_6$};
	\node (s7) [right of=sx,e] {$\sigma_7$};
	\node (s8) [right of=s7,e] {$\sigma_8$};
	\path (s1.east) edge[->,gray] node[above,pos=.45,gray]{$\beta_1$} (s2.west)
		  (s2.east) edge[->] 		node[above,pos=.45](a){$\alpha_1$} (s3.west)
		  (s3.east) edge[->,gray] node[above,pos=.45,gray](e){$\beta_2$} (s4.west)
		  (s4.east) edge[->,gray] node[above,pos=.45,gray]{$\beta_3$} (s6.west)
%		  (s5.east) edge[->] 		node[above,pos=.45]{$\alpha_2$} (s6.west)
		  (s6.east) edge[->] 		node[above,pos=.45]{$\alpha_2$} (sx.west)
		  (sx.east) edge[->,gray] node[above,pos=.45,gray]{$\beta_4$} (s7.west)
		  (s7.east) edge[->] 		node[above,pos=.45]{$\alpha_3$} (s8.west);

	\node (s1a) [node distance=.75cm,below of=s1,e] {$\sigma_1$};
	\node (s2a) [right of=s1a,e] {$\sigma_2$};
	\node (s3a) [right of=s2a,e] {$\sigma_3'$};
	\node (s4a) [right of=s3a,e] {$\sigma_4$};
%	\node (s5a) [right of=s4a,e] {};
	\node (s6a) [right of=s4a,e] {$\sigma_5$};
	\node (sxa) [right of=s6a,e] {$\sigma_6$};
	\node (s7a) [right of=sxa,e] {$\sigma_7$};
	\node (s8a) [right of=s7a,e] {$\sigma_8$};
	\path (s1a.east) edge[->,gray] node[above,gray]{$\beta_1$} (s2a.west)
		  (s2a.east) edge[->,gray] node[above,gray](ea){$\beta_2$} (s3a.west)
		  (s3a.east) edge[->] 		node[above](aa){\vphantom{$\beta$}$\alpha_1$} (s4a.west)
		  (s4a.east) edge[->,gray] node[above,gray](eea){$\beta_3$} (s6a.west)
%		  (s5a.east) edge[->] 		node[above]{$\alpha_2$} (s6a.west)
		  (s6a.east) edge[->] 		node[above,pos=.45]{$\alpha_2$} (sxa.west)
		  (sxa.east) edge[->,gray] node[above,pos=.45,gray]{$\beta_4$} (s7a.west)
		  (s7a.east) edge[->] 		node[above,pos=.45]{$\alpha_3$} (s8a.west);

	\path  (a.south) edge[->,thick,dotted] node[pos=.45]{} (aa.north)
		   (e.south) edge[->,thick,dotted] node[pos=.45]{} (ea.north);

	\node (s1b) [node distance=.75cm,below of=s1a,e] {$\sigma_1$};
	\node (s2b) [right of=s1b,e] {$\sigma_2$};
	\node (s3b) [right of=s2b,e] {$\sigma_3'$};
	\node (s4b) [right of=s3b,e] {$\sigma_4'$};
%	\node (s5b) [right of=s4b,e] {};
	\node (s6b) [right of=s4b,e] {$\sigma_5$};
	\node (sxb) [right of=s6b,e] {$\sigma_6$};
	\node (s7b) [right of=sxb,e] {$\sigma_7$};
	\node (s8b) [right of=s7b,e] {$\sigma_8$};

	\path (s1b.east) edge[->,gray] node[above,gray]{$\beta_1$} (s2b.west)
		  (s2b.east) edge[->,gray] node[above,gray]{$\beta_2$} (s3b.west)
		  (s3b.east) edge[->,gray] node[above,gray](eb){$\beta_3$} (s4b.west)
		  (s4b.east) edge[->] node(ab)[above]{\vphantom{$\beta$}$\alpha_1$} (s6b.west)
%		  (s5b.east) edge[->] node[above]{$\alpha_2$} (s6b.west)
		  (s6b.east) edge[->] 		node(b2)[above,pos=.45]{$\alpha_2$} (sxb.west)
		  (sxb.east) edge[->,gray] node(b1)[above,pos=.45,gray]{$\beta_4$} (s7b.west)
		  (s7b.east) edge[->] 		node(b2)[above,pos=.45]{$\alpha_3$} (s8b.west);

	\path (eea.south) edge[->,thick,dotted] node[pos=.45]{} (eb.north)
		  (aa.south) edge[->,thick,dotted] node[pos=.45]{} (ab.north);

	\node (s1c) [node distance=.75cm,below of=s1b,e] {$\sigma_1$};
	\node (s2c) [right of=s1c,e] {$\sigma_2$};
	\node (s3c) [right of=s2c,e] {$\sigma_3'$};
	\node (s4c) [right of=s3c,e] {$\sigma_4'$};
%	\node (s5c) [right of=s4c,e] {};
%	\node (s6c) [right of=s4c,e] {$\sigma_5$};
%	\node (sxc) [right of=s6c,e] {$\sigma_6$};
	\node (s7c) [node distance=.75cm,below of=s7b,e] {$\sigma_7'$};
	\node (s8c) [right of=s7c,e] {$\sigma_8$};

	\path (s1c.east) edge[->,gray] node[above,gray]{$\beta_1$} (s2c.west)
		  (s2c.east) edge[->,gray] node[above,gray]{$\beta_2$} (s3c.west)
		  (s3c.east) edge[->,gray] node[above,gray](eb){$\beta_3$} (s4c.west)
		  (s4c.east) edge[->] node(ab)[above,pos=.1]{$\alpha_1$} 
%		  (s5c.east) edge[->] node[above]{$\alpha_2$} (s6c.west)
							node(c1)[above,pos=.3]{$\circ$}
							node(c1)[above,pos=.5]{$\alpha_2$}
							node(c1)[above,pos=.7]{$\circ$}
		  					node(c1)[above,pos=.9]{\vphantom{$\beta$}$\alpha_3$} (s7c.west)
		  (s7c.east) edge[->,gray] 		node(c2)[above,pos=.45,gray]{$\beta_4$} (s8c.west);

	\path (b1.south) edge[->,thick,dotted] node[pos=.45]{} (c2.north)
		  (b2.south) edge[->,thick,dotted] node[pos=.45]{} (c1.north);
\end{tikzpicture}
}\vspace{-2.5em}
\end{wrapfigure}
The example (right) shows the evolution of a trace when a
reduction with $n\hspace{-1mm}=\hspace{-1mm}3$, $k\hspace{-1mm}=\hspace{-1mm}2$ is applied.
Actions $\beta_1,\ldots,\beta_4$ are remote.
The pre-action $\alpha_1$ is first moved towards the commit action
$\alpha_2$. Then the same is done with the post-action $\alpha_3$.
\textbf{L1} resp. \textbf{L2} guarantee that the trace's end state $\sigma_8$
remains invariant, \textbf{L3} guarantees its existence and
\textbf{L4} guarantees that e.g.
$\sigma_4 \notin\varphi \implies
\sigma_3' \notin\varphi$ and
$\sigma_6
\notin\varphi \implies
 \sigma_7' \notin\varphi$
(preserving invariant violations $\neg\varphi$ in the reduced system without $\sigma_4$ and $\sigma_6$).
%The internal states (see \autoref{sec:introduction}) can be split in
%pre- and post-commit states.
%The (pre-commit
%%The action $\alpha_k$ might interact with other threads and therefore
%%is called the \emph{commit} in the
%%database terminology~\cite{papadimitriou}.
%%Actions preceding it are called \concept{pre-commit} actions and gather resources, such as locks.
%%The remaining  actions are \concept{post-commit} actions that (should) release these resources.
%%Hence, the beginning of an atomic section can only be extended (upwards) with right movers,
%%and the end of an atomic section can only be extended (downwards) with left movers. 
%%(or are typically supposed to).
%We refer to pre(/post)-commit transitions including source and target states
%as the \concept{pre(/post) phase}.
The subsequent section provides a dynamic variant of TR.

%!TEX root = main.tex

\newcommand\trp{\stackrel{}{\longrightarrow_{i}'}}
\defmath\phase{\mathit{h}}

\section{Stubborn Transaction Reduction}
\label{sec:str}

The current section gradually introduces stubborn transaction reduction.
First, we introduce a stubborn set definition that is parametrized with
different commutativity relations.
In order to have enough luggage to compare POR to TR in \autoref{sec:comparison},
we elaborate here on various aspects of stubborn POR and compare
our definitions to the original stubborn set definitions.
%Comparing the parametrized definition to the original, we show that
%only strong commutativity yields valid reductions, while left-commutativity requires
%additional constraints. 
We then provide a definition for \concept{dynamic left and right movers}, based on
the  stubborn set parametrized with left and right commutativity.
Finally, we provide a definition of a \concept{transaction system},
show how it is reduced and provide an algorithm to do so.
This demonstrates that TR can be made dynamic in the same sense as stubborn
sets are dynamic. %, while also being able to use left and right commutativity.
We focus in the current paper on the preservation of invariants.
But since deadlock preservation is an integral part of POR, it is addressed
as well.

\vspace{-.5em}
\subsection{Parametrized stubborn sets}\label{s:pss}
\vspace{-.5em}
We use stubborn sets as they have advantages compared to other
traditional POR techniques~\cite[Sec.~4]{intuition}.
%and can also be extended to preserve LTL and CTL~\cite{ssalgebra}.
We first focus on a basic definition of the stubborn set that
only preserves deadlocks.
The following version is parametrized (with~$\star$).

\begin{definition}[$\star$-stubborn sets]\label{def:ss}
Let $\star\in \set{\shortleftarrow, \shortrightarrow, \leftrightarrow}$.
A set $B \subseteq \actions$ is $\star$-stubborn in the state $\sigma$, written $\st^\star_\sigma(B)$, if:
\begin{provisos}
\item \label{i:d0}
			$\en(\sigma) \neq \emptyset \implies B\cap \en(\sigma) \neq \emptyset$
			\hfill (include an enabled action, if one exists)
\item \label{i:d2}
			$\forall \alpha \in B \cap\overline{\en}(\sigma)\colon
			\exists E\in\nes_\sigma(\alpha) \colon  E\subseteq B$
			\phantom{x}\hfill(for disabled $\alpha$ include a \concept{NES})
\item \label{i:d1}\vspace{-1mm}
			$\forall \alpha\in B \cap \en(\sigma), \beta \stackrel{\star}{\not\bowtie} \alpha \colon
			\beta\in B$
			  %\vspace{-1mm}\stackrel{\star}{\not\bowtie}(\alpha) \subseteq B$
			\hfill(for enabled $\alpha$ include $\star$-dependent actions)
\end{provisos}
%\qed
\end{definition}

Notice that a stubborn set $B$ includes actions disabled in $\sigma$
to reason over future behavior with \textbf{\ref{i:d2}}: 
Actions $\alpha\in B$ commute with $\beta\in \en(\sigma)\setminus B$ by~\textbf{\ref{i:d1}},
but also with $\beta'\in \en(\sigma')$ for $\sigma\tr{\beta}{}\sigma'$, since
\textbf{\ref{i:d2}} ensures that $\beta$ cannot enable any $\gamma\in B$ (ergo $\beta'\notin B$).
\autoref{th:valmari} formalizes this.
From $B$, 
the reduced system is obtained~by taking $\por(\sigma) \defn \en(\sigma)\cap B$:
It preserves deadlocks.
But not all $\star$-parametrizations lead to correct reductions w.r.t. deadlock preservation.
We therefore briefly relate our definition to the original stubborn set definitions.
The above definition yields three interpretations of a set $B\subseteq \actions$ for
a state $\sigma$.
\begin{itemize}
\item If $\st^\leftrightarrow_\sigma(B)$, then $B$ coincides with the original \emph{strong} stubborn
						set~\cite{valmari1988error,parle89}. %with $\por(\sigma) \defn \en(\sigma)\cap B$.
\item If $\st^\leftarrow_\sigma(B)$, then $B$ approaches the weak stubborn set in~\cite{guardpor2},
	   a simplified version of~\cite{valmari1}, except that it
		lacks a necessary \emph{key action} (from \cite[Def.~1.17]{valmari1}).\footnote{
\textbf{\ref{i:d0}} is generally not preserved
with left-commutativity ($\star = \leftarrow$), as $\beta \notin B$ may disable $\alpha\in B$.
Consequently, $\beta$ may lead to a deadlock.
Because POR prunes all $\beta \notin B$,
$\st^\leftarrow_\sigma(B)$ is not a valid reduction (it may prune deadlocks).
The key action repairs this by demanding at least one \emph{key action} $\alpha$,
which strongly commutes, i.e.,  $\forall \beta\in B \colon \alpha \comm \beta$,
which by virtue of strong commutativity cannot be disabled by any $\beta\notin B$. 
}
\item If $\st^\rightarrow_\sigma(B)$, then $B$ also may yield an invalid POR, as
	it would consider two locking operations independent and thus potentially miss a deadlock.
\end{itemize}

This indicates that POR, unlike TR, cannot benefit from right-commutativity.
The consequences of this difference are further discussed in \autoref{sec:comparison}.
The strong version of our bare-bone stubborn set definition, on the other hand,
is equivalent to the one presented~\cite{valmari1} and thus
preserves the `stubbornness' property~(\autoref{th:valmari}).
If we define \concept{semi-stubbornness}, written $\sst^\star_\sigma$,
like stubbornness minus the \textbf{\ref{i:d0}} requirement,
%and with one additional requirement for only $\sst^\star_\sigma$:
%\begin{provisosp}
%%  \setcounter{provisosi}{0}
%\item \label{i:d1p}\vspace{-1mm}
%			$\forall i,\alpha,\alpha'\in B \cap \en(\sigma)\cap\actions_i,
%			\beta\notin B\colon 
%			c_\alpha \lrestr \tr{\beta}{} \rrestr \overline{c_\alpha} \cap
%			\overline{c_\alpha} \lrestr \tr{\beta}{} \rrestr c_\alpha \neq\emptyset
%			\implies
%			\beta\in B$\\
%			(if $\alpha$ is non-detministic (with $\alpha'$), then
%			 $\beta \notin B$ must either dis- or enable it)
%\end{provisosp}
%
then we can prove a similar theorem for semi-stubborn sets~(\autoref{th:sst}).\footnote{
We will show that semi-stubbornness, i.e., $\sst^\shortleftarrow_\sigma(B)$ (without key),
is sufficient for stubborn TR, which may therefore prune deadlocks.
Contrarily, invariant-preserving stubborn POR is strictly stronger
than the basic stubborn set (see below), and hence also preserves all deadlocks.
(This is relevant for the POR/TR comparison in \autoref{sec:comparison}.)
}
%For simplicity, we limit ourselves to \concept{action-deterministic} stubborn sets in
%\autoref{th:sst} here. However, \autoref{app:proofs} provides a proof for a more general version
%of the theorem that supports non-determinism between stubborn actions.
This `stubbornness' of semi-$\shortleftarrow$ and semi-$\shortrightarrow$ stubborn sets
is used below to define dynamic movers. First, we briefly return our attention to
stubborn POR, recalling how it preserves properties beyond deadlocks and the computation of~$\st_\sigma$.

\vspace{-.5em}
\begin{theorem}[\cite{valmari1}]%stub sets for reduced state space generation~
\label{th:valmari}
If $B\subseteq \actions$, $\st^{\leftrightarrow}_\sigma(B)$ and $\sigma\tr{\beta}{}\sigma'$
for $\beta\notin B$, then
$\st^{\leftrightarrow}_{\sigma'}(B)$.
\vspace{-.5em}
\end{theorem}
%!TEX root = ../main.tex
\begin{theorem}\label{th:sst}%stub sets for reduced state space generation
If $B\subseteq \actions$, $\sst^{\star}_\sigma(B)$ and $\sigma \tr{\beta}{}\sigma'$
for $\beta\notin B$,
then $\sst^{\star}_{\sigma'}(B)$ for $\star \in \set{\shortleftarrow,\leftrightarrow}$,
as well as for $\star \in \set{\shortrightarrow}$
provided that $\beta$ does not disable
a stubborn action, i.e., $\en(\sigma) \cap B \subseteq \en(\sigma') \cap B$.
%,
%and $B$ is action-deterministic in $\sigma$, i.e., $\sizeof{B\cap \en(\sigma)}= 1$.
\end{theorem}
%\begin{theorem}\label{th:sst}%stub sets for reduced state space generation
%If $B\subseteq \actions$, $\sst^{\shortleftarrow}_\sigma(B)$ and $\sigma \tr{\beta}{}\sigma'$
%for $\beta\notin B$,
%then $\sst^{\shortleftarrow}_{\sigma'}(B)$.
%\end{theorem}
\vspace{-.5em}

\paragraph{Stubborn sets for safety properties}
To preserve a safety property such as $\Box \varphi$ (i.e. $\varphi$ is invariant),
a stubborn set $B$
($\st^{\leftrightarrow}_\sigma(B) = \true$) needs to satisfy two additional
requirements~\cite{explosion} called \textbf{S} for \concept{safety} and
\textbf{V} for \concept{visibility}.
% \defn \sigma \tr{\alpha}{} \sigma' \land \alpha\in\por(\sigma)$.
To express \textbf{V}, we denote actions enabling $\varphi$ with
$\actions_\oplus^\varphi$
and those disabling the proposition with $\actions_\ominus^\varphi$.
Those combined form the visible actions:
$\actions^\varphi_{\mathit{vis}} \defn \actions_\ominus^{\varphi}\cup \actions_\oplus^{\varphi}$.
For \textbf{S}, recall that
$\sigma \xdashrightarrow\alpha \sigma'$ is a reduced transition.
\concept{Ignoring states} disrespect \textbf{S}.

\begin{itemize}
\item[\textbf S] $\forall \beta\hspace{-.8mm}\in\hspace{-.8mm}\en(\sigma) \colon \exists \sigma' \colon 
							\sigma \portrans^{\hspace{-1mm}*}\hspace{-1mm}\sigma' \land \beta\hspace{-.8mm}\in\hspace{-.8mm}\por(\sigma')$
\hfill
			\textit{(never keep ignoring pruned actions)}
\item[\textbf V] $B\cap\en(\sigma)\cap\actions_{\mathit{vis}}^\varphi \neq \emptyset \implies \actions_{\mathit{vis}}^\varphi \subseteq B$
\hfill
			\textit{(either all or no visible, enabled actions)}
\end{itemize}

\paragraph{Computing stubborn sets and heuritics}
POR is not deterministic as we may compute many different valid stubborn 
sets for the same state and we can even select different ignoring states
to enforce the \textbf{S} proviso
(i.e. the state $\sigma'$ in the \textbf{S} condition).
A general approach to obtain good reductions is to compute a
stubborn set with the fewest enabled actions, so that the $\por(\sigma)$
set is the smallest and the most actions are pruned in $\sigma$.
However, this does not necessarily lead to the best reductions as
observed several times~\cite{explosion,fairtesting,varpa}.
Nonetheless, this is the best heuristic currently available, and it
generally yields good results~\cite{guardpor2}.

The $\forall\exists$-recursive structure of Def.~\ref{def:ss} indicates that
establishing the smallest stubborn set is an \NP-complete problem,
which indeed it is~\cite{optimal}.
Various algorithms exist to heuristically compute small stubborn
sets~\cite{guardpor2,deletion}. Only the deletion
algorithm~\cite{deletion} provides guarantees on the returned sets
(that no strict subset of the return set is also stubborn). On the other hand, the
guard-based approach~\cite{guardpor2} has been shown to deliver good reductions in reasonable
time.

To implement the \textbf{S} proviso,
Valmari~\cite{valmari1} provides an algorithm~\cite[Alg.~1.18]{valmari1}
that yields the fewest possible ignoring states, runs in linear time and
can even be performed on-the-fly,
i.e. while generating the reduced transition system. It is based
on Tarjan's strongly connected component (SCC) algorithm~\cite{tarjan}.

The above methods are relevant for stubborn TR
%, presented in a subsequent subsection,
as STR also needs to compute ($\star$-)stubborn sets and
avoid ignoring states (recall \text{L3} from \autoref{sec:prelim}).

\subsection{Reduced transaction systems}

TR merges sequential actions into (atomic) transactions and
in the process removes interleavings (at the states internal to the transaction) just like POR.
We present a dynamic TR that decides to prolong transactions on a per-state~basis.
We use stubborn sets to identify left and right moving actions in each state.
Unlike stubborn set POR, and much like ample-set POR~\cite{katz-peled},
we rely on the process-based action decomposition to identify sequential parts of the system.

\newcommand\rem{\hspace{-.6mm}}
Recall that actions in the pre-phase should commute to the right and actions in the post-phase should
commute to the left with other threads. We use the notion of stubborn sets
to define \concept{dynamic left and right movers} in \autoref{eq:left} and \ref{eq:exclude}
for $\tuple{\sigma,\alpha,\sigma'}\in\trans_i $.
%In the case of right-movers, we only require the \textbf{\ref{i:d2}} conditions from semi-stubborn sets.
%We call a set $B\subseteq \actions$ which satisfies only \textbf{\ref{i:d2}} with $\star=\rightarrow$ in a state~$\sigma$ 
%\concept{feebly stubborn}, denoted $\fst^\rightarrow_\sigma(B)$.
Both mover definitions are based on semi-stubborn sets.
Dynamic left movers are state-based requiring all outgoing local transitions to ``move'',
whereas right movers are action-based allowing different reductions for various non-deterministic paths.
%Also, right movers are not allowed to disable non-stubborn actions,
%i.e. $\en(\sigma)\rem\setminus\rem B \subseteq  \en(\sigma')\rem\setminus\rem B$,
%but this is only required if there is an $\alpha'\in\actions_i$ 
%non-deterministic with $\alpha$ as \autoref{app:proofs} shows.
The other technicalities of the definitions stem from the different premises of
left and right movability (see \autoref{sec:prelim}).
Finally, both dynamic movers exhibit a  crucial monotonicity property, similar to
previously introduced `stubbornness', as expressed by \autoref{lem:dlm} and \autoref{lem:drm}.
\vspace{-1ex}
%!TEX root = ../main.tex
\begin{IEEEeqnarray}{lll}\label{def:m}
\rem\rem\rem\rem\rem\rem\rem\rem\rem\rem\rem
\lmv_i(\sigma) &\defn %\{ \sigma\in\states \mid&
				 &\exists B \colon
		\sst_{\sigma}^\shortleftarrow(B),~
					B \cap \en(\sigma) = \actions_i \cap \en(\sigma)
		\label{eq:left}\\*
\rem\rem\rem\rem\rem\rem\rem\rem\rem\rem\rem
\rmv_i({\sigma,\alpha,\sigma'})
			&\defn %\{ \tuple{\sigma,\alpha,\sigma'}\rem \in T_i \mid
			&\exists B \colon
		\sst_{\sigma}^\shortrightarrow(B),~\alpha\in B,~
		 B\cap \en(\sigma') \subseteq \actions_i,~
		 B\cap \en(\sigma) = \set{\alpha}%,~
%		 \en(\sigma)\rem\setminus\rem B \subseteq  \en(\sigma')\rem\setminus\rem B~
		  %\hspace{1em}(, \tuple{\sigma,\alpha,\sigma'}\in\trans_i)
	\label{eq:exclude}
\vspace{-1ex}
\end{IEEEeqnarray}

%!TEX root = ../main.tex
\begin{lemma}\label{lem:dlm}
The dynamic left-moving property is never
remotely disabled, i.e.:
if $\lmv_i(\sigma_1) \land i\neq j \land \sigma_1\tr{\beta}{j} \sigma_2$,
then $\lmv_i(\sigma_2)$.
\end{lemma}\vspace{-3ex}
%!TEX root = ../main.tex
\begin{lemma}\label{lem:drm}
Dynamic right-movers retain dynamic moveability after moving, i.e.:
if $\rmv_i(\sigma_1,\alpha,\sigma_2) \land \sigma_1\tr{\alpha}{i} \sigma_2 \tr{\beta}{j} \sigma_3$ for $i\neq j $,
then $\exists \sigma_1 \tr{\beta}{j}\sigma_4 \colon \rmv_i(\sigma_4,\alpha,\sigma_3)$.
\end{lemma}

To establish stubborn TR, 
\autoref{def:trs} first annotates the transition system \cts
with thread-local phase information, i.e. one phase variable for each thread
that is only modified by that thread.
%Unlike POR, TR does not preserve deadlocks. We focus here on simple safety properties:
%invariants and show that TR can in fact both remove and introduce deadlocks to
%obtain better reductions (see also \autoref{sec:comparison}).
Phases are denoted with \NN (for transaction external states),
\RR (for for states in the pre-phase) and \LL (for states in the post phase).
Because phases now depend on the commutativity established via
dynamic movers \autoref{eq:left} and \ref{eq:exclude},
the reduction (not included in the definition, but discussed below it) becomes dynamic.
\autoref{lem:preserves} follows easily as the definition does not yet enforce the reduction, but mostly `decorates' the transition system.

\begin{definition}[Transaction system]\label{def:trs}
Let $H \defn \set{\NN, \RR, \LL} ^ P$ be an array of local phases.
The transaction system is CTS $\cts'\defn \tuple{\tstates, \ttrans, \actions, \tsint}$ such that:
\vspace{-2ex}
%!TEX root = ../main.tex
\begin{IEEEeqnarray}{rllrl}
\hspace{-5ex}\tstates 	\defn &\parbox{0cm}{\mbox{$\states \times \phases,~~~~~~~~~~
					\tsint 	\defn\tuple{\sint, \NN^{{\procs}}}	$}}				\nonumber\\*
%\hspace{-5ex}\tsint 	\defn &\parbox{0cm}{\mbox{$\tuple{\sint, \NN^{{\procs}}}	$}}		\nonumber\\*
\hspace{-5ex}\ttrans_i 	\defn &\parbox{0cm}{\mbox{$\{ \tuple{\tuple{\sigma,h},\alpha,\tuple{\sigma',h'}} \in{\tstates}\times\actions\times\tstates \mid
				(\sigma,\alpha,\sigma')\in T_i
				,\,\,\, \forall j\neq i\colon h_j'=h_j,$}}			\nonumber\\*
%\hspace{-5ex} &\parbox{0cm}{\mbox{~~~~~~~~~~~~~~~~~~~~~~~~~
%				$\forall j\colon \phase_j = \LL \implies \lmv_j(\sigma)\land \en(\sigma)\cap\actions_j\neq \emptyset,$}}\\*
&& \RR &\hphantom{|XX}& \text{iff } h_i \neq \LL
						\,\,\,\land\,\,\,
						\rmv_i(\sigma,\alpha,\sigma')
						\,\,\,\land\,\,\,
						\alpha\notin \actions_\ominus^\varphi			\label{eq:1} \\*
&\smash{h_i'\hspace{-.5ex}=\hspace{-.5ex}\left\{\IEEEstrut[8\jot] \right.}
& \LL && \text{if } \lmv_i(\sigma')
						\,\,\,\land\,\,\, 
						\en(\sigma')\cap\actions_i\cap \actions_\oplus^\varphi=\emptyset	\label{eq:2}\\*
&& \NN && \text{otherwise (or as alternative when \autoref{eq:2} holds)} 				\label{eq:3}\\*
&  \parbox{0cm}{\mbox{$\}$%, with  %$\rmv_i\defn \bigcup_{\alpha\in\actions_i} \rmv_\alpha$.}}
%$\rmv_i\defn \set{(\sigma,\sigma')\mid \bigcup_{\alpha\in\actions_i}(\sigma,\alpha,\sigma')\in\rmv_\alpha}$.
}}
%\text{, and~}\ttrans \defn \biguplus_i \ttrans_i
%,\text{ where } \LL_i \defn \set{ \tuple{\sigma,\phase} \mid \phase_i = \LL }$}}
\nonumber
\vspace{-2ex}
\end{IEEEeqnarray}

\noindent
%$\textsf{X}_i \defn \set{ \tuple{\sigma,\phase} \mid \phase_i = \textsf{X} }$
%for $\textsf{X} \in \set{\LL, \RR, \NN}$.

\end{definition}
\begin{lemma}\label{lem:preserves}
\autoref{def:trs} preserves invariants:
$\reach(\cts) \models\Box\varphi \Leftrightarrow\reach(\cts')\models \Box\varphi$.
\end{lemma}

The conditions in \autoref{eq:2} and \autoref{eq:3} overlap on purpose,
allowing us to enforce termination below.
The transaction system effectively partitions the state spaces on the
phase for each thread $i$, i.e.
 $\overline{\NN_i} = \LL_i\cup\RR_i$ with 
$\NN_i \defn \set{ \tuple{\sigma,\phase} \mid \phase_i = \textsf{\NN} }, \text{etc}$.
%We write $\NN_i, \RR_i, \LL_i$ for $\set{ \tuple{\sigma,\phase} \mid \phase_i = \textsf{\NN} }$,~etc.
The definition of $\ttrans_i$ further ensures three properties:

\begin{enumerate}[label=\Alph*.]
\item\label{i:A} $\LL_i$ states do not transit to $\RR_i$ states as $h_i=\LL \implies h_i'\neq \RR$ by \autoref{eq:1}.
\item Transitions ending in $\RR_i$ are dynamic right movers not disabling
		$\varphi$~by~\autoref{eq:1}.
\item Transitions starting in $\LL_i$ are dynamic left movers not enabling
		$\varphi$~by~\autoref{eq:2}.
%\item $\LL_i$ states terminate by reaching an $\NN_i$ state
%	as $ \vcs_i(\LL_i)  \implies  h_i'= \NN$ by \autoref{eq:3}.
%\item\label{i:D}  $\LL_i$ states can be changed to $\NN_i$ states by \autoref{eq:2}. The
%	reason behind this will become apparent below.
\end{enumerate}

\noindent
Thereby $\ttrans_i$ implements the (syntactic) constraints from Lipton's
TR (see \autoref{sec:prelim})
dynamically in the transition system, except for \textbf{L3}.
Let $\trp \defn \set{\tuple{q,q'} \mid \tuple{q,\alpha,q'} \in T_i' }$. 
Next, \autoref{th:reduction} defines the reduced transaction system (RTS),
based primarily on the $\trtrans$ transition relation that only allows
a thread $i$ to transit when all other threads are in an external state,
thus eliminating interleavings
(\brtrans additionally skips internal states).
The theorem concludes that invariants are preserved given that a
termination criterium weaker than \textbf{L3} is met:
All $\LL_i$ must reach~an
~$\NN_i$~state. Monotonicity of dynamic movers plays a key role in its proof.

%!TEX root = ../main.tex
\begin{theorem}[Reduced Transaction System (RTS)]\label{th:reduction} We define for all $i$:
\vspace{-2ex}
\begin{IEEEeqnarray*}{llr}
%\NN_i &\defn \set{ \tuple{\sigma,\phase} \mid \phase_i = \textsf{\NN} }, \text{etc} &
%\text{\hspace{-2cm}(state partitioning: $\overline{\NN_i} = \LL_i\cup\RR_i$)}\\
%Let $q \overarrow{\alpha}_i' q'$ for $\tuple{q,\alpha,q'}\in \ttrans_i$.\\
%$W \defn \bigcup_i W_i$, $N \defn \bigcap_i N_i$ and 
%$N_\tx \defn \bigcap_{u\neq t} N_j$.\\
\trtrans_i &\defn (\cup_{j\neq i} \NN_j) \lrestr \trp
					&\text{($i$  only transits when all $j$ are external)}\\
\brtrans_i &\defn \NN_i\lrestr (\trtrans_i\rrestr \overline{\NN_i})^* \trtrans_i \rrestr \NN_i
&\text{(skip internal states transition relation)}%\\
%\brtrans &\defn \bigcup_i \brtrans_i  &\\
%\stackrel\brtrans{\scriptsize \cts} &\defn \langle \tstates,\{\tuple{q,\alpha_i,q'}\mid
%                                        q\stackrel\alpha\brtrans_i q' \},\actions,\tsint \rangle)
%                                        &\text{(the reduced transaction system)}
\vspace{-2ex}
\end{IEEEeqnarray*}
The RTS is a CST
$\stackrel\brtrans{\scriptsize \cts} \defn\rem\rem\rem \langle \tstates,\{\tuple{q,\alpha_i,q'}\mid
                                        q\stackrel{\alpha_i}\brtrans_i q' \},\actions,\tsint \rangle$.
Now, provided that
$\forall \sigma \in \LL_i \colon \exists \sigma'\in\NN_i \colon \sigma\hookrightarrow^*_i \sigma' $,
we have
$\reach(\cts')\models \Box\varphi \Longleftrightarrow
\reach(\stackrel\brtrans{\scriptsize \cts})\models \Box\varphi$.
\end{theorem}

\algrenewcommand\algorithmicprocedure{\textbf{proc}}
\algrenewcommand\alglinenumber[1]{\scriptsize #1:}

\begin{algorithm}[t]
\caption{Algorithm reducing a CTS to an RTS using  $T_i'$ from \autoref{def:trs}.}
\label{alg:rtrs}
%{\fontsize{7pt}{7pt}\selectfont
\begin{minipage}{0.5\textwidth}
\begin{algorithmic}[1]  
\State $V_1,V_2,Q_1,Q_2\colon \states'$
   
%\Procedure{Explore}{$\cts\defn \tuple{\states,\trans,\actions,\sint}$}
%	\AssignState{$Q_1$}{$\set{ \tuple{\sint, \NN^\procs} }$}
%	\AssignState{$V_1$}{$\emptyset$}
%	\State \textsc{Search}($\cts$)
%	\State \textbf{assert}($V_1 = \reach(\cts^\brtrans)$)
%\EndProcedure

\Procedure{Search}{$\cts\defn \tuple{\states,\trans,\actions,\sint}$}
\AssignState{$Q_1$}{$\set{ \tuple{\sint, \NN^\procs} }$}
\AssignState{$V_1$}{$\emptyset$}
\While{$Q_1\neq \emptyset$}
	\AssignState{$Q_1$}{$Q_1 \setminus \set{\tuple{\sigma,\phase}}$ \textbf{for} $\tuple{\sigma,\phase}\in Q_1$}	
	\AssignState{$V_1$}{$V_1 \cup \set{\tuple{\sigma,\phase}}$}
	\State \textbf{assert}($\forall i \colon \phase_i = \NN $)	\label{l:assert}
	\For{$i \in P$}
		\State \textsc{Transaction}(T, $\tuple{\sigma,\phase}$, i)				
	\EndFor
\EndWhile
\State \textbf{assert}($V_1 = \reach(\brcts)$)
\EndProcedure
\Function{SCCroot}{$q, i$}
	\State \textbf{return }$q$ is a root of bottom SCC $C$
	\phantom{XxXXXX} \textbf{s.t.} $C \subseteq \LL_i \land C \subseteq V_2$
\EndFunction.

\algstore{alg2}
\end{algorithmic}
\end{minipage}
\hfill
\begin{minipage}{0.53\textwidth}
\begin{algorithmic}[1]
\algrestore{alg2}

\Procedure{Transaction}{$T$, \tuple{\sigma,\phase}, $i$}
	\AssignState{$Q_2$}{$\set{ \tuple{\sigma, \phase} }$}
	\AssignState{$V_2$}{$\emptyset$}
\While{$Q_2\neq \emptyset$}
	\AssignState{$Q_2$}{$Q_2 \setminus \set{\tuple{\sigma,\phase}}$ \textbf{for} $\tuple{\sigma,\phase}
		\in Q_2$}	
	\AssignState{$V_2$}{$V_2 \cup \set{\tuple{\sigma,\phase}}$}
\For{$\tuple{\sigma,\alpha,\sigma'} \in T_i$}
%		\AssignState{$\phase'$}{\textbf{according to \autoref{def:trs}}}
    	\State{\textbf{let} $h'$ \textbf{s.t.} $\tuple{\tuple{\sigma,h}\hspace{-.5mm},\alpha,\tuple{\sigma'\hspace{-.5mm},h'}}\in T'_i$\phantom{X}}%(\autoref{def:trs})}
%		\State \textbf{if }$\phase'_i = \NN$ \textbf{ then }$\phase_i'$ := $\RR$
%		\If{$\phase'_i = \RR \land
%			(\tuple{\sigma,\alpha,\sigma'}\notin M^\rightarrow_\alpha
%			\lor \alpha\in \actions_\ominus^\varphi)$}
%		\AssignState{$\phase'_i$}{$\LL$}
%		\EndIf
%		\If{$\phase'_i\hspace{-0mm}=\hspace{-0mm}\LL \land
%				(\set{\sigma'}\hspace{-1mm}\lrestr\hspace{-.5mm} T_i 
%						\hspace{-.5mm}\not\subseteq\hspace{-.5mm} M^\leftarrow_i
%				\hspace{-.0mm}\lor\hspace{-.0mm}
%					\en_i(\sigma')\hspace{-.5mm}\cap\hspace{-.5mm}\actions_\oplus^\varphi
%					\hspace{-.5mm}\neq\hspace{-.5mm}\emptyset)$}
%		\AssignState{$\phase'_i$}{$\NN$}
%		\EndIf
	\If{$\textsc{SCCroot}(\tuple{\sigma',h'}, i)$}
	\AssignState{$\phase'_i$}{$\NN$}\label{l:ext}
	\EndIf
	\If{$\tuple{\sigma'\hspace{-.5mm},h'}\not\sqsubseteq V_1 \cup V_2 \cup Q_1 \cup Q_2$}\label{l:sub2}
		\AssignState{$Q_2$}{$Q_2 \cup \set{\tuple{\sigma',h'}}$}
	\EndIf
	\If{$\phase'_i = \NN \land \tuple{\sigma'\hspace{-.5mm},h'} \notin V_1\cup Q_1$}\label{l:propagate}
	\AssignState{$Q_1$}{$Q_1 \cup \set{\tuple{\sigma',h'}}$}
	\EndIf
\EndFor
\EndWhile
\EndProcedure
\end{algorithmic}
\end{minipage}
%}
\end{algorithm}

The following algorithm generates the RTS~~$\brcts$ of \autoref{th:reduction}
from a \cts.
%the reduced system with
%\brtrans-transitions as defined in \autoref{th:reduction}.
The state space search is split into two:
One main search, which only processes external states ($\bigcap_i \NN_i$), and an additional
search (\textsc{Transaction}) which explores the transaction for a single thread $i$.
Only when the transaction search encounters an external state, it is
propagated back to the queue $Q_1$ of the main search, provided it is new there
(not yet in $V_1$, which is checked at Line~\ref{l:propagate}).
The transaction search terminates early when an internal state $q$ is found to
be subsumed by an external state already encountered in the outer search
(see the $q\not\sqsubseteq V_1$ check at Line~\ref{l:sub2}).
Subsumption is induced by the following order on phases, which is
lifted to states and sets of states $X\subseteq \states'$:
$\RR \sqsubset \LL \sqsubset \NN$ with $a \sqsubseteq b \Leftrightarrow a = b \lor a \sqsubset b$,
$\tuple{\sigma,\rem\phase} \rem\sqsubseteq\rem \tuple{\sigma',\rem\phase'} \Leftrightarrow \sigma = \sigma'\land \forall i\colon \phase_i \rem\sqsubseteq\rem \phase'_i$, and
$q\rem\sqsubseteq\rem X\Leftrightarrow \forall q'\rem\rem\in\rem\rem X \colon q\rem\sqsubseteq\rem q'$
(for $q = \tuple{\sigma,\rem \phase}$).

Termination detection is implemented using Tarjan's SCC algorithm as in~\cite{valmari1}.
We chose not to obfuscate the search with
the rather intricate details of that algorithm.
Instead, we assume that there is a function \textsc{SCCRoot}
which identifies a unique \concept{root} state in each bottom SCC
composed solely of post-states. This state is then
made external on Line~\ref{l:ext} fulfilling the premise of \autoref{th:reduction}
($\forall \sigma \in \LL_i \colon \exists \sigma'\in\NN_i \colon \sigma\hookrightarrow^*_i \sigma' $).
Combined with \autoref{lem:preserves} this yields \autoref{th:alg}.

\begin{theorem}\label{th:alg}\vspace{-.2em}
\autoref{alg:rtrs} computes $\reach(\brcts)$ s.t.
$\reach(\cts)\models \Box\varphi \Longleftrightarrow
\reach(\brcts)\models \Box\varphi$.
%$\forall \sigma \in \LL_i \colon \exists \sigma'\in\NN_i \colon \sigma\hookrightarrow^*_i \sigma' $.
\end{theorem}

Finally, while the transaction system exponentially blows up the number of
syntactic states ($\neq$ reachable states) by adding local phase variables, the reduction completely hides this
complexity as \autoref{th:removal} shows.
Therefore, as soon as the reduction succeeds in removing a single state,
we have by definition that 
$\sizeof{\reach(\cts)} <  \lvert\reach(\brcts)\rvert$.
\autoref{th:removal} also allows us to simplify the algorithm
by storing transition system states $\states$ instead
of transaction system states
\tstates in $V_1$ and $Q_1$.
\begin{theorem}\label{th:removal}\vspace{-.8em}
Let $N\defn \cap_{i} \NN_i$.
We have $\sizeof{N} = \sizeof{S}$ and
$\reach(\brcts)\subseteq \reach(\cts)$.
%by \autoref{th:reduction}.
\vspace{-.5em}
\end{theorem}

%!TEX root = main.tex

\section{Comparison between TR and POR}
\label{sec:comparison}

Stubborn TR (STR) %presented in the previous subsection
is dynamic in the same sense as stubborn POR,
allowing for a better comparison of the two.
To this end, we discuss various example types of systems
that either TR or POR excel at.
As a basis, consider a
completely independent system with $p$ threads of $n-1$ operations each.
Its state space has $n^p$ states.
TR can reduce a state space to $2^p$ states whereas POR yields $n * p$ states.
The question is however also which kinds of systems are realistic
and whether the reductions can be computed precisely and efficiently.

%\subsection{Cases where POR obtains better reductions}

%\subsection{Cases where TR obtains better reductions}

\paragraph{High parallelism vs Long sequences of local transitions}
%As shown by \autoref{f:lipton}.
%In cases where the input model performs many sequential steps, TR will perform well.
POR has an advantage
when $p \gg n$ being able to yield exponential reductions.
Though e.g. thread-modular verification~\cite{Flanagan2002,Malkis2006}
may become more attractive in those cases.
Software verification often has to deal with many sequential actions benefitting STR,
especially when VM languages such as LLVM are used~\cite{vvt}.

\begin{wrapfigure}[5]{r}{1.5cm}\vspace{-2.5em}
\hspace{-.2cm}
\begin{tikzpicture}
   \tikzstyle{e}=[minimum width=1cm]
   \tikzstyle{every node}=[font=\small, node distance=1.5cm]

  \node (s1) {$l_0$};
  \node (s2) [above right of=s1,yshift=-.3cm] {$l_1$};
  \node (s2p) [above right of=s1,yshift=-.8cm] {$..$};
  \node (s3) [below right of=s1,yshift=.3cm] {$l_9$};
  \node (s3p) [below right of=s1,yshift=.8cm] {$..$};

  \path (s1) edge[->] node[sloped,pos=.48]{} (s2);
  \path (s1) edge[->,dashed] node[sloped,pos=.48]{} (s2p);
  \path (s1) edge[->,dashed] node[sloped,pos=.48]{} (s3p);
  \path (s1) edge[->] node[sloped,pos=.48]{} (s3);

  \path (s2) edge[bend left=-30,->] node[sloped,pos=.48]{} (s1);
  \path (s3) edge[bend left=30,->] node[sloped,pos=.48]{} (s1);
\end{tikzpicture}
\end{wrapfigure}
\paragraph{Non-determinism}
In the pre-phase, TR is able to individually reduce mutually non-deterministic
transitions of one thread due to \autoref{eq:1}, which contrary to \autoref{eq:2}
considers individual actions of a thread.
Consider the example on the right. It represents a system
with nine non-determinisitic steps in a loop.
Assume one of them never commutes, but the others commute to the right.
Stubborn TR is able to reduce all paths through the loop over only
the right-movers, even if they constantly yield new states (and interleavings).
%{\scriptsize
%\begin{verbatim}
%proc A() {
%	int x = 0, a = 0, i = 1;
%	do
%	:: i < 10; x++; i++; % make non-commuting?
%	:: i < 10; a+=1; i++;
%	....
%	:: i < 10; a+=10; i++;
%	:: else -> break;
%	od;
%	x -= 9;
%	a -= 10;
%	assert (x > 0 && a > 0);
%}
%proc B() {
%	do
%	:: x = __pid;
%	od;
%}
%\end{verbatim}
%}

\begin{wrapfigure}[15]{r}{4.6cm}\vspace{-2.5em}
\scalebox{.75}{\hspace{-1.7em}
%!TEX root = ../main.tex

\begin{tikzpicture}[node distance=1.3cm,font=\smaller]
\newcommand\xstate[3]{(\tuple{#1,#2},#3)}
\renewcommand\xstate[3]{\tuple{#3}}

\tikzset{snake/.style={
decoration={snake, 
    amplitude = .4mm,
    segment length = 2mm,
    post length=0.9mm},decorate}}

%   \tikzstyle{every node}=[minimum width=1cm]
\tikzstyle{e}=[line width=2pt]

    \node (s0x) {};
    \node (s1) [below left	of=s0x] 	{};%$\xstate{2}{a}{0,0}$};
    \node (s2) [below right of=s0x] 	{};%$\xstate{1}{b}{0,1}$};
    \node (s3) [below left 	of=s1] 	{};%$\xstate{3}{a}{0,2}$};
    \node (s4) [below right of=s1]	{};%$\xstate{2}{b}{0,1}$};
    \node (s5) [below right of=s2] 	{};%$\xstate{1}{c}{1,1}$};
    \node (s6) [below right	of=s3] 	{};%$\xstate{3}{b}{0,2}$};
    \node (s7) [below right	of=s4] 	{};%$\xstate{2}{c}{1,1}$};
    \node (s8) [below left	of=s7] 	{};%$\xstate{2}{c}{1,1}$};

	\node (sa) [above right	of=s0x] 	{};%$\xstate{2}{a}{0,0}$};
    \node (sb) [above right of=s2] 	{};%$\xstate{1}{b}{0,1}$};
    \node (sc) [above right	of=s5] 	{};%$\xstate{3}{a}{0,2}$};

	\node (sx1) [below right	of=sc] 	{};%$\xstate{2}{a}{0,0}$};
    \node (sx2) [below right of=s5] 	{};%$\xstate{1}{b}{0,1}$};
    \node (sx3) [below right	of=s7] 	{};%$\xstate{3}{a}{0,2}$};
    \node (sx4) [below right	of=s8] 	{};%$\xstate{3}{a}{0,2}$};

\def\TONE{}
\def\TTWO{}
\def\sone{\ccode{P(m);}}
\def\stwo{\ccode{P(m);}}

    \path (s0x) edge[->,e] node(t1')[midway,above,sloped] {\ccode{x=1;}}
    					node(t1')[midway,below,sloped] {$\TONE$}  (s1);
    \path (s0x) edge[->,e] node(t1')[midway,above,sloped] {\stwo}
    					node(t1')[midway,below,sloped] {$\TTWO$} (s2);
    \path (s1) edge[->,e] node(t1')[midway,above,sloped] {\ccode{V(m);}}
    					node(t1')[midway,below,sloped] {$\TONE$} (s3);
    \path (s1) edge[->] node(t1')[midway,above,sloped] {\stwo}
					    node(t1')[midway,below,sloped] {$\TTWO$} (s4);
    \path (s2) edge[->,e] node(t1')[midway,above,sloped] {\ccode{x=1;}}
					    node(t1')[midway,below,sloped] {$\TONE$} (s4);
    \path (s2) edge[->,e] node(t1')[midway,above,sloped] {\ccode{x=2;}}
					    node(t1')[midway,below,sloped] {$\TTWO$} (s5);
    \path (s4) edge[->] node(t1')[midway,above,sloped] {\ccode{V(m);}}
					    node(t1')[midway,below,sloped] {$\TONE$} (s6);
    \path (s4) edge[->,e] node(t1')[midway,above,sloped] {\ccode{x=2;}}
    					node(t1')[midway,below,sloped] {$\TTWO$} (s7);
    \path (s3) edge[->,e] node(t1')[midway,above,sloped] {\stwo}
					    node(t1')[midway,below,sloped] {$\TTWO$} (s6);
    \path (s5) edge[->,e] node(t1')[midway,above,sloped] {\ccode{x=1;}}
    					node(t1')[midway,below,sloped] {$\TONE$} (s7);
    \path (s7) edge[->,e] node(t1')[midway,above,sloped] {\ccode{V(m);}}
					    node(t1')[midway,below,sloped] {$\TONE$} (s8);
    \path (s6) edge[->,e] node(t1')[midway,above,sloped] {\ccode{x=2;}}
    					node(t1')[midway,below,sloped] {$\TTWO$} (s8);

    \path (sa) edge[->,e] node(t1')[midway,above,sloped] {\sone}
    					node(t1')[midway,below,sloped] {$\TONE$} (s0x);
    \path (sb) edge[->,e] node(t1')[midway,above,sloped] {\sone}
					    node(t1')[midway,below,sloped] {$\TONE$} (s2);
    \path (sc) edge[->] node(t1')[midway,above,sloped] {\sone}
    					node(t1')[midway,below,sloped] {$\TONE$} (s5);
    \path (sa) edge[->,e] node(t1')[midway,above,sloped] {\stwo}
    					node(t1')[midway,below,sloped] {$\TTWO$} (sb);
    \path (sb) edge[->,e] node(t1')[midway,above,sloped] {\ccode{x=2;}}
					    node(t1')[midway,below,sloped] {$\TTWO$} (sc);

    \path (sc) edge[->,e] node(t1')[midway,above,sloped] {\ccode{V(m);}}
    					node(t1')[midway,below,sloped] {$\TONE$} (sx1);
    \path (s5) edge[->] node(t1')[midway,above,sloped] {\ccode{V(m);}}
    					node(t1')[midway,below,sloped] {$\TONE$} (sx2);
    \path (s7) edge[->] node(t1')[midway,above,sloped] {\ccode{V(m);}}
    					node(t1')[midway,below,sloped] {$\TONE$} (sx3);
    \path (s8) edge[->,e] node(t1')[midway,above,sloped] {\ccode{V(m);}}
    					node(t1')[midway,below,sloped] {$\TONE$} (sx4);

    \path (sx1) edge[->,e] node(t1')[midway,above,sloped] {\ccode{l(m);}}
    					node(t1')[midway,below,sloped] {$\TONE$} (sx2);
    \path (sx2) edge[->,e] node(t1')[midway,above,sloped] {\ccode{x=1;}}
    					node(t1')[midway,below,sloped] {$\TONE$} (sx3);
    \path (sx3) edge[->,e] node(t1')[midway,above,sloped] {\ccode{V(m);}}
    					node(t1')[midway,below,sloped] {$\TONE$} (sx4);

    \path (sx1) edge[->,e,dashed,bend left=65] node(t1')[midway,below,sloped] {\sone\ccode{x=1;V(m);}} (sx4);
    \path (sa) edge[->,e,dashed,bend right=65] node(n1)[midway,above,sloped] {\sone\ccode{x=1;V(m);}} (s3);
    \path (sa) edge[->,e,dashed,bend right=-65] node(t1')[midway,above,sloped] {\stwo\ccode{x=2;V(m);}} (sx1);
    \path (s3) edge[->,e,dashed,bend right=65] node(t1')[midway,below,sloped] {\stwo\ccode{x=2;V(m);}} (sx4);

\end{tikzpicture}
}
\vspace{-2.5em}
\caption{State space of \ccode{P(m); x=1; V(m); $\|$ P(m);  x=2; V(m);}
		 and 
		POR (thick lines) and TR (dashed lines).
%(c.f. \cite{peled-93,parle89,godefroid}).
}
\label{f:locks}
\end{wrapfigure}
\paragraph{Left and right movers}
While stubborn POR can handle left-commutativity using additional restrictions,
STR can benefit from right-commutativity in the pre-phase and 
from left-commutativity in the post-phase.
E.g., P/V-semaphores are right/left-movers (see \autoref{sec:prelim}).
\autoref{f:locks} shows a system with ideal reduction using TR,
and none with stubborn set POR.

\autoref{tab:sync} provides various synchronization constructs and their movability.
Thread create \& join have not been classified before.

\begin{table}[t]
\caption{Movability of commonly used synchronization mechanisms}\label{tab:sync}
\smaller
\begin{tabular}{l|p{9.5cm}}
\toprule
%\\	\midrule
\texttt{pthread\_create} & As this can be modeled with a mutex that is guarding the thread's code
and is initially set to locked, the \texttt{create}-call is an unlock and thus a left-mover.   \\
\texttt{pthread\_join} & Using locking similar to \texttt{create},
				\texttt{join} becomes a lock and thus a
				right-mover.   \\
Re-entrant locks & Right / left movers~\cite{qadeer-transactions}\\
Wait/notify/notifyAll & Can all three be split into right and left moving parts~\cite{qadeer-transactions}\\
\bottomrule
\end{tabular}\vspace{-1.5em}
\end{table}

\paragraph{Deadlocks}
POR preserves all deadlocks, even when irrelevant to the property.
TR does not preserve deadlocks at all,
potentially allowing for better reductions preserving invariants.
The following example deadlocks because of an invalid locking order. 
TR can still reduce the example to four states, creating maximal transactions.
%for each thread.
On the other hand, POR must explore the deadlock.

{\centering
$\footnotesize
~~~~~~~\ccode{l(m1);l(m2); x=1; u(m1);u(m2); $\|$ l(m2);l(m1);  x=2; u(m1);u(m2);}
$}

%Encoding deadlocks as invariant for use in TR would
%limit reductions.
%When encoding deadlocks as an invariant property, 
%TR pulls in almost all strict right and left movers, as they
%might respectively enable or disable the property (absence of a deadlock).
%Consequently, few reduction can be expected.
%It seems however possible to weaken the visibility requirement in TR,
%but this remains an open problem.

\paragraph{Processes}
STR retains the process-based definition from its
ancestors~\cite{lipton}, while stubborn POR can go beyond
process boundaries to improve reductions and even supports
process algebras~\cite{ssalgebra,guardpor2}.
%It remains an open problem to rid STR of its process-based definition.
In early attempts to solve the open problem of 
a process-less STR definition, we observed that 
inclusion of all actions in a transaction
could cause the entire state space search to
move to the \textsc{SearchTransaction} function.

\paragraph{Tractability and heuristics}
The STR algorithm can fix the set of stubborn
transitions to those in the same thread
(see definitions of $M_\alpha^\star$).
This can be exploited in the deletion algorithm by fixing the relevant
transitions (see the incomplete minimization approach~\cite{deletion}).
If the algorithm returns a set with other transitions, then we know
that no transaction reduction is possible as the returned set is
subset-minimal~\cite[Th.~1]{guardpor2}.
The deletion algorithm runs in polynomial time
(in the order of $\sizeof{\actions}^4$~\cite{explosion}),
hence also stubborn TR also does (on a per-state basis).
Stubborn set POR, however, is \NP-complete as it
has to consider all subsets of actions.
Moreover, a small stubborn set is merely a heuristic
for optimal reductions~\cite{optimal} as discussed in \autoref{s:pss}.

%\subsection{Known unknowns}
%The following aspects of both techniques require further comparison.

\paragraph{Known unknowns}
We did not consider other properties such as full safety,  LTL and CTL.
For CTL, POR can no longer reduce to non-trivial subsets
because of the CTL proviso~\cite{gerth1995partial}
(see \cite{ssalgebra} for support of non-deterministic transitions,
like in stubborn TR).
TR for CTL is an open problem.

%\paragraph{Ignoring \& visibility}
While TR can split visibility in enabling (in the pre-phase) and
disabling (in the post-phase), POR must consider both combined.
POR moreover must compute the ignoring proviso over the entire state space
while TR only needs to consider post-phases and thread-local steps.
%It is unclear which has the advantage in these cases.

%\paragraph{Parallelizations}
The ignoring proviso~\cite{valmari-92,ignoring,openset}
in POR tightly couples the possible reductions per state
to the role the state plays in the entire reachability graph.
This lack of locality adds an extra obstacle to the parallelization of the
model checking procedure. Early results make compromises in the obtained
reductions~\cite{parpor}. Recent results show that reductions do not
have to be affected negatively even with high amounts of
parallelism~\cite{cycleproviso}, however these results have not
yet been achieved for distributed systems.
TR reduction on the other hand, offers plenty of parallellization
opportunities, as each state in the out search can be handed off to
a separate process.

%!TEX root = main.tex

\newcommand\spin{\textsc{SPIN}\xspace}
\newcommand\ltsmin{\textsc{LTSmin}\xspace}

\section{Experiments}
\label{sec:experiments}

We implemented stubborn transaction reduction (STR) of \autoref{alg:rtrs} in the open source
model checker \ltsmin\footnote{\url{http://fmt.cs.utwente.nl/tools/ltsmin/}}~\cite{ltsmin},
using a modified deletion algorithm to establish optimal stubborn sets
in polynomial time (as discussed in \autoref{sec:comparison}).
The implementation can be found on GitHub.\footnote{\url{https://github.com/alaarman/ltsmin/commits/tr}}
%\ltsmin is a high-performance language-independent model checker,
%which provides enumerative (explicit-state) and symbolic (BDD-based) algorithms
%for state space exploration and safety / LTL / CTL checking.
%These algorithms are implemented in parallel, but we run them here sequentially
%to keep as many variables equal when comparing against the POR algorithms.
\ltsmin
has a front-end for \textsc{promela} models, which is on par with
the \spin model checker~\cite{holzmann-97} performance-wise~\cite{spins}.
Unlike SPIN, \textsc{LTSmin} does not implement
dynamic commutativity specifically for queues~\cite{holzmann-peled},
but because it splits queue actions into a separate action for each
cell~\cite{spins}, a similar result is achieved by virtue of the
stubborn set condition \textbf{\ref{i:d2}} in \autoref{s:pss}.
This benefits both its POR and STR implementation.

%we compare against the SPIN model checker. TODO
%Both \ltsmin and \spin accepts models written in \textsc{promela}.
\begin{wraptable}{r}{6cm}
\vspace{-1em}
\caption{Models and their verification times in \ltsmin. Time in sec. and memory use~in~MB.
State/transition counts are the same in both \ltsmin and \spin.}
\label{tab:models}
\setlength{\tabcolsep}{.5ex} % row width
{
\scriptsize
\begin{tabular}{l|rr|rr}
\toprule
&	\multicolumn{2}{c|}{\textsc{SPIN/LTSmin}}
&	\multicolumn{2}{c}{\ltsmin}
\\
&	states
&	transitions
&	time
&	mem.
\\	\midrule
\texttt{Peterson5}			& 829909270 & 3788955584 	& 4201. 	& 6556.  	\\
\texttt{GARP}			& 48363145 	& 247135869 	& 88.34 	& 369.8 	\\
\texttt{i-Prot.2}		& 13168183 	& 44202271 	& 22.99 	&  102.8  	\\
\texttt{i-Prot.0}		& 9798465 	& 45932747 	& 19.58 	& 75.2 	\\
\texttt{Peterson4}		& 3624214 	& 13150952 	& 7.36 	& 28.5 \\
\texttt{BRP}			& 2812740 	& 6166206 	& 4.59 	& 26.4 	\\
\texttt{MSQ}			& 994819	& 3198531 	& 4.41 	& 12.1  	\\
\texttt{i-Prot.3}		& 327358 	& 978579 	& 0.79 	& 2.8 	\\
\texttt{i-Prot.4}		& 78977 	& 169177 	& 0.19 	& 0.8 	\\
\texttt{Small1}			& 36970 	& 163058 	& 0.14 	& 0.3 	\\
%\texttt{Peterson.3}		& 17523 	& 47004 	& 0.05 	& 0.2 	\\
\texttt{X.509}			& 9028 	& 35999 	& 0.03 	& 0.1 	\\
\texttt{Small2}			& 7496 	& 32276 	& 0.08 	& 0.1 	\\
\texttt{SMCS}			& 2909 	& 10627 	& 0.01 	& 0.1 	\\
%\texttt{Logical}		& 1026 	& 3715 	& 0.010 	& n/a 	\\
%\texttt{SmallSTA}		& 876 	& 2147 	& 0.00 	& n/a 	\\
\bottomrule
\end{tabular}
\vspace{-2em}
}
\end{wraptable}
\begin{table*}[b!]\vspace{-4ex}
\caption{Reduction runs of TR, Stubborn TR (STR) and Stubborn POR (SPOR).
Reductions of states $|\states|$  and transitions $|\trans|$ are given in
percentages (reduced state space / original state space),
runtimes in sec. and memory use in MB.
The lowest reductions (in number of states) and the runtimes are highlighted in bold.}
\label{tab:tr}
\setlength{\tabcolsep}{.35ex} % row width
\smaller
{
\begin{tabular}{l|rrrr|rrrr|rrrr|rrrr}
\toprule
&	\multicolumn{4}{c|}{TR (\ltsmin)} 
&	\multicolumn{4}{c|}{\textbf{STR} (\ltsmin)} 	
&	\multicolumn{4}{c|}{SPOR (\ltsmin)} 	
&	\multicolumn{4}{c}{Ampe set (\spin)} 
\\
&	$|\states|$
&	$|\trans|$
&	time
&	mem

&	$|\states|$
&	$|\trans|$
&	time
&	mem

&	$|\states|$
&	$|\trans|$
&	time
&	m.

&	$|\states|$
&	$|\trans|$
&	time
&	mem
\\	\midrule
%\texttt{GARP}&	100&	100&	307.69&	369.8&	1,4&	1,5&	776.53&	5.2&	3,6&	1,5&	33.13&	13.5&	7,6&	3,7&	6.27&	289.1\\
%\texttt{i-Prot.0}&	100&	100&	68.15&	75.2&	30,3&	31,5&	323.36&	22.9&	32,1&	17,2&	81.99&	24.3&	15,7&	10,5&	2.56&	132.2\\
%\texttt{Peterson.4}&	1,3&	1,0&	0.47&	0.5&	1,3&	1,0&	1.23&	0.5&	7,3&	2,7&	1.91&	2.4&	14,7&	6,8&	0.24&	28.9\\
%\texttt{BRP}&	100&	100&	10.36&	26.4&	47,6&	36,9&	6.91&	12.6&	100&	100&	32.96&	26.4&	9,2&	6,0&	0.18&	22.2\\
%\texttt{i-Prot.3}&	31,3&	29,4&	0.59&	1.0&	12,2&	11,3&	1.05&	0.4&	20,7&	10,4&	0.34&	0.6&	27,0&	16,5&	0.06&	5.8\\
%\texttt{i-Prot.4}&	71,9&	73,3&	0.34&	0.6&	34,1&	36,1&	0.46&	0.3&	45,2&	31,5&	0.18&	0.4&	50,4&	37,1&	0.03&	2.8\\
%\texttt{Small1}&	8,9&	18,0&	0.05&	n/a&	6,7&	13,6&	0.08&	n/a&	31,2&	17,7&	0.06&	0.1&	48,4&	45,1&	0.01&	0.9\\
%\texttt{Peterson.3}&	4,2&	3,0&	0.10&	n/a&	4,2&	3,0&	0.02&	n/a&	15,9&	7,4&	0.01&	n/a&	17,1&	8,1&	n/a&	0.3\\
%\texttt{X.509}&	100&	100&	0.06&	0.1&	19,3&	16,8&	0.05&	n/a&	7,8&	3,7&	0.02&	n/a&	67,5&	34,3&	0.01&	1.1\\
%\texttt{Small2}&	11,6&	21,0&	0.00&	n/a&	8,7&	15,7&	0.03&	n/a&	35,0&	19,8&	0.02&	n/a&	48,3&	43,8&	0.01&	0.4\\
%\texttt{SMCS}&	100&	100&	0.06&	0.1&	25,6&	19,2&	0.11&	n/a&	12,5&	5,3&	0.04&	n/a&	41,1&	19,6&	n/a&	0.7\\
%\texttt{Logical}&	12,7&	51,8&	0.01&	n/a&	12,6&	51,7&	0.00&	n/a&	100&	75,9&	0.01&	n/a&	100&	100&	n/a&	0.3\\
%\texttt{SmallSTA}&	36,4&	61,0&	0.00&	n/a&	27,2&	49,9&	0.00&	n/a&	75,1&	63,2&	0.01&	n/a&	50,5&	50,1&	n/a&	0.3\\

\texttt{Peterson5}&	0.5&	0.3&	\bf6.11&	 33.0 &	\bf 0.4&	0.3&	74.01&	 29.5 &	3.1&	0.9&	316.10&	 209.8 &	5.2&	1.9&	42.30 &	2463. \\
\texttt{GARP}&	100&	100&	266.21&	 369.8 &	\bf1.4&	1.5&	776.53&	5.2&	3.6&	1.5&	19.83&	 13.5 &	7.6&	3.7&	\bf6.27&	289.1\\
\texttt{i-Prot.2}&	\bf2.1&	2.4&	\bf3.46&	 2.2 &	\bf2.1&	2.4&	4.87&	 2.2 &	20.2&	11.9&	13.32&	 21.7 &	26.1&	17.6&	4.33 &	246.9 \\
\texttt{i-Prot.0}&	100&	100&	56.71&	 75.2 &	\bf12.8&	12.5&	148.78&	 9.7 &	32.1&	17.2&	214.93&	 24.3 &	15.7&	10.5&	\bf2.56&	132.2\\
\texttt{Peterson4}&	\bf1.3&	1.0&	 0.36 &	 0.5 &	\bf1.3&	1.0&	 0.85 &	 0.5 &	7.3&	2.7&	4.24&	 2.4 &	14.7&	6.8&	\bf0.24&	28.9\\
\texttt{BRP}&	100&	100&	9.59&	 26.4 &	47.6&	36.9&	6.38&	 12.6 &	100&	100&	90.31&	 26.4 &	\bf9.2&	6.0&	\bf0.18&	22.2\\
\texttt{MSQ}&	66.0&	65.0&	5.5&	 8.2 &	\bf22.9&	21.5&	14.90&	 3.0 &	52.1&	29.1&	12.14&	 6.5 &	80.4&	46.6&	\bf1.03 &	200.9 \\
\texttt{i-Prot.3}&	8.0&	7.4&	 0.19&	 0.2 &	\bf8.0&	7.4&	0.24&	 0.2 &	20.7&	10.4&	 0.94&	 0.6 &	27.0&	16.5&	\bf0.06&	5.8\\
\texttt{i-Prot.4}&	25.1&	27.2&	 0.14&	 0.2 &	\bf25.0&	27.1&	0.18&	 0.2 &	45.2&	31.5&	 0.54&	 0.4 &	50.4&	37.1&	\bf0.03&	2.8\\
\texttt{Small1}&	8.9&	18.0&	 0.03&	n/a&	\bf6.7&	13.6&	0.07&	n/a&	31.2&	17.7&	 0.18&	 0.1 &	48.4&	45.1&	\bf0.01&	0.9\\
%\texttt{Peterson.3}&	4.2&	3.0&	n/a&	n/a&	4.2&	3.0&	0.01&	n/a&	15.9&	7.4&	 0.03&	n/a&	17.1&	8.1&	n/a&	0.3\\
\texttt{X.509}&	93.8&	94.1&	 0.07&	 0.1 &	19.3&	16.7&	0.06&	n/a&	\bf7.8&	3.7&	 0.03&	n/a&	67.5&	34.3&	\bf0.01&	1.1\\
\texttt{Small2}&	11.6&	21.0&	 0.01&	n/a&	\bf8.7&	15.8&	0.01&	n/a&	35.0&	19.8&	 0.04&	n/a&	48.3&	43.8&	\bf0.01&	0.4\\
\texttt{SMCS}&	100&	100&	 0.05&	 0.1 &	26.1&	19.6&	0.09&	n/a&	\bf12.5&	5.3&	 0.03&	n/a&	41.1&	19.6&	\bf0.01&	0.7\\

%LOGICAL.spinout	&1026	&4357	&0	& 0.287	&\\
%SMALL1.spinout	&17887	&73545	&0.01	& 0.9
%SMALL2.spinout	&3619	&14137	&0.01	& 0.4	&\\
%SMALLSTA.spinout	&442	&1075	&0	& 0.3	&\\
%X.509.spinout	&6094	&12336	&0.01	& 1.1
%brp.spinout	&259750	&372848	&0.18	& 22.2
%garp.spinout	
%i0.spinout	&1538288	&4843498	&2.56	& 132.2
%i3.spinout	&88233	&161673	&0.06	& 5.8
%i4.spinout	&39789	&62799	&0.03	& 2.8
%peterson3.spinout	&2999	&3805	&0	& 0.3	&\\
%peterson4.spinout	&533083	&888657	&0.24	& 28.9
%smcs.spinout	&1196	&2086	&0	& 0.7	&\\

\bottomrule
\end{tabular}
}
\end{table*}

We compare STR against (static) TR from \autoref{sec:prelim}.
We also compare STR against the stubborn set POR  in \ltsmin, which was
shown to consistently outperform SPIN's ample set~\cite{holzmann-peled} implementation
%(which contains a language-specific form of )
in terms of
reductions, but with worse runtimes due to the more elaborate stubborn set
algorithms (a factor 2--4)~\cite{guardpor2}.
(We cannot compare with \cite{vmcai} due to the different input
formats of VVT~\cite{vvt} and \textsc{LTSmin}.)
\autoref{tab:models} shows the
models that we considered and their normal (unreduced) verification times in \ltsmin.
We took all models from \cite{guardpor2} that contained an assertion.
The inputs include mutual exclusion algorithms (\texttt{peterson}),
protocol implementations (\texttt{i-protocol}, \texttt{BRP, GARP, X509}),
a lockless queue (\texttt{MSQ}) and controllers (\texttt{SMCS}, \texttt{SMALL1}, \texttt{SMALL2}).

\ltsmin runs~with~STR were configured according to the command line:\\
\texttt{\scriptsize prom2lts-mc --por=str --timeout=3600 -n --action=assert m.spins}\\
The option \texttt{--por=tr} enables the static TR instead.
We also run all models in \spin in order to compare against the ample
set's performance. \spin runs were configured according to the following command lines:\\
\texttt{\scriptsize
cc -O3 -DNOFAIR -DREDUCE -DNOBOUNDCHECK -DNOCOLLAPSE %\textbackslash\\\phantom{x}~~~~~~~~
-DSAFETY -DMEMLIM=100000 -o pan pan.c\\
./pan -m10000000 -c0 -n -w20 
}

\autoref{tab:tr} shows the benchmark results.
We observe that STR often surpasses POR (stubborn and ample sets) in terms of reductions.
Its runtimes however are inferior to those of the ample set in SPIN.
This is likely because we use the precise deletion algorithm, which decides the optimal
reduction for STR: STR is the only algorithm of the four that does not use heuristics.
The higher runtimes of STR are often compensated by the better reductions it obtains.

Only three models demonstrate that POR can yield better reductions
(\texttt{BRP}, \texttt{smcs} and \texttt{X.509}).
This is perhaps not surprising as these models do not have massive parallelism
(see \autoref{sec:comparison}). It is however interesting to note that
\texttt{GARP} contains seven threads.
We attribute the good reductions of STR mostly to its ability to skip internal states.
\spin's ample set only reduces the \texttt{BRP} better than \ltsmin's stubborn
POR and STR. 
In this case, we found that \ltsmin too eagerly identifies half of the actions
of both models as visible.

\paragraph{Validation}
Validation of TR is harder than of POR. For POR, we usually count deadlocks,
as all are preserved, but TR might actually prune deadlocks and error states
(while preserving the invariant as per \autoref{th:alg}).
We therefore tested correctness of our implementation by implementing methods that check
the validity of the returned semi-sturbborn sets. Additionally, we maintained counters
for the length of the returned transactions and
inspected the inputs to confirm validity of the longest transactions.

%TODO: disabple ignoring proviso

%!TEX root = main.tex

\section{Related Work} %and Future TODO
\label{sec:related}

Lipton's reduction was refined multiple times
\cite{lamport-lipton,gribomont,Doeppner:1977:PPC:512950.512965,lamport-tla,Stoller2003}.
Flanagan et al.~\cite{Flanagan2002,Flanagan:2003:TA:640136.604176} and Qadeer et al.~\cite{qadeer-atomicity,qadeer-java,qadeer-transactions}
have most recently developed transactions and found various applications.
The reduction theorem used to prove the theorems in the current paper
comes from our previous work~\cite{vmcai}, which in turn is
a generalized version of \cite{qadeer-transactions}. 
Our generalization allows the direct support of dynamic transactions as already
demonstrated for symbolic model checking with IC3 in~\cite{vmcai}.
Despite a weaker theorem, Qadeer and Flanagan~\cite{qadeer-transactions} 
can also dynamically grow transactions by doing iterative refinement 
over the state space exploration.
This contrasts our approach, which instead allows on-the-fly adaptation of
movability (within a single exploration).
Moreover, \cite{qadeer-transactions} bases dynamic behavior
on exclusive access to variables, whereas our technique can handle any kind of
dependency captured by the general stubborn set POR relations.

Cartesian POR~\cite{cartesian} is a form of
Lipton reduction that builds transactions during the exploration,
%yet still uses static dependence relations.
but does not exploit left/right commutativity.
%TODO
The leap set method~\cite{schoot} treats disjoint
reduced sets in the same state as truly concurrent and executes
them as such: The product of the different disjoint sets is 
executed from the state, which entails that sequences of
actions are executed from the state. This is where the similarity with the
TR ends, because in TR the sequences are formed by
sequential actions, whereas in leap sets they consist of concurrent actions,
e.g., actions from different processes.
Recently, trace theory has been generalized to include `steps' by
Ryszard et al.~\cite{Janicki2016}. We believe that this work could form a basis
to study leap sets and TR in more detail.

%This instrumentation further allows non-deterministic edges to 
%have separate phase transits as
%the first row in the table can by applied
%for non-deterministic edges $(l_a,\alpha,l_b)$ and $(l_a,\alpha',l_{b'})$.
%This is avoided in the post-phase
%(the third row in the table combines all outgoing transitions),
%because if one edge is a dynamic non-mover its source location must become
%external and the other edges might just as well benefit from that
%(by resetting the transaction).

Various classical POR works were mentioned, e.g.~\cite{parle89,godefroid,katz-peled}.
How `persistent sets'~\cite{godefroid}/`ample sets'~\cite{katz-peled}
relate to stubborn set POR is explained in~\cite[Sec.~4]{intuition}.
Sleep sets~\cite{godefroid-wolper-93} form an orthogonal approach, but in isolation
only reduce the number of transitions.
Dwyer et al.~\cite{dwyer2004exploiting} propose dynamic techniques for object-oriented programs.
Completely dynamic approaches exist~\cite{flanagan2005dynamic,Kastenberg2008}.
Recently, even optimal solutions were found~\cite{abdullah,sousa,contextpor}.
These approaches are typically stateless however, although still
succeed in pruning converging paths sometimes (e.g.,~\cite{sousa}).
Others aim at making dependency more dynamic~\cite{godefroid-pirottin,holzmann-peled,collapses}.
%Note that stubborn POR achieves a similar result
%by virtue of \textbf{\ref{i:d2}} depending on how the guards are implemented.

%Comparisons of symbolic with enumerative POR techniques have not been performed to the best of our knowledge.
Symbolic POR can be more static for reasons discussed in Footnote~\ref{fn:symbolic}, e.g.,~\cite{alur1997partial}.
%Therefore, Grumberg et al.~\cite{Grumberg:2005:PUM:1040305.1040316} \& Wang et al.~\cite{wang} present
%property-guided symbolic methods.
Therefore, Grumberg et al.~\cite{Grumberg:2005:PUM:1040305.1040316} present
underapproximation-widening, which iteratively refines an under-approximated
encoding of the system. In their implementation, interleavings are constrained
to achieve the under-approximation. Because refinement is done based on verification
proofs, irrelevant interleavings will never be considered.
%The technique currently only supports BMC and the implementation is not available, so we did not compare against it.
Other relevant dynamic approaches are peephole and monotonic POR by
Wang et al.~\cite{peephole,kahlon-wang-gupta}. 
Like sleep sets~\cite{godefroid}, however, these methods only reduce
the number of transitions. While a reducing transitions can speed up
symbolic approaches by constraining the transition relation,
it is not useful for enumerative model checking, which
is strongly limited by the amount of unique states that need to be stored in memory.
%TODO: check for monotonic

Kahlon et al.~\cite{Kahlon2006} do not implement transactions, but
encode POR for symbolic model checking using SAT. The ``sensitivity'' to locks of their algorithm
can be captured in traditional stubborn sets as well by viewing locks as normal ``objects'' (variables)
with guards, resulting in the subsumption of the ``might-be-the-first-to-interfere-modulo-lock-acquisition''
relation~\cite{Kahlon2006} by the ``might-be-the-first-to-interfere'' relation~\cite{Kahlon2006},
originally from~\cite{godefroid}.
%The latter in turn is a language-specific, object oriented and stronger version of the
%necessary enabling set (NES), introduced in \autoref{sec:prelim}.

Elmas et al.~\cite{ElmasQT09} propose
dynamic reductions for type systems, where the 
invariant is used to weaken the mover definition.
They also support both right and left movers, but do 
automated theorem proving instead of model checking.

%Our previous work~\cite{vmcai} supports more dynamic behavior and comes with a lean encoding
%because of its use of `dynamic mover conditions'.

% dynamic
% - Flanagan, abdullah, kroening (also reduces action sequences)
% - peephole, monotonic (sleepsets)

% Fully dynamic via transition splitting

% Gribomont, french guy, Cohen,  Cohen and L. Lamport, Qadeer - us
% leaping														<-------
% covering step graphs 											<-------
% POR: Valmari, Peled, guard-based
% sleep sets: Godefroid
% leaping
% qadeer has shown many applications for reductions
% remote related: Grumberg

% process algebras
% - arend
%!TEX root = main.tex

\section{Conclusion} %and Future TODO
\label{sec:conclusion}
We presented a more dynamic version of transaction reduction (TR) based
on techniques from stubborn set POR.
We analyzed several scenarios for which either of the two approaches
has an advantage and also experimentally compared both techniques.
We conclude that TR is a valuable alternative to
POR at least for systems with a relatively low amount of parallelism. 

Both in theory and practice, TR showed advantages to POR, 
but vice versa as well.
Most strikingly, TR is able to exploit various synchronization mechanisms 
in typical parallel programs because of their left and right commutativity.
While not preserving deadlocks, its reductions can benefit from omitting them.
These observations are supported by experiments that show
better reductions than a comparably dynamic POR
approach for systems with up to 7 threads.
We observe that the combination POR and TR is an open problem.

%TODO: discuss process-based nature

%\section{Acknowledgements}
%This work is partially
%supported by the Austrian National Research Network S11403-N23 (RiSE)
%of the Austrian Science Fund (FWF) and by the Vienna Science and
%Technology Fund (WWTF) through grant~VRG11-005.

\bibliographystyle{plain}
\bibliography{lit}

\clearpage
%!TEX root = main.tex
 
\appendix

\section{Correctness Proofs}
\label{app:proofs}

The current appendix contains the proofs for the lemmas and theorems in the paper.
For clarity, lemmas and theorems are repeated with the same numbering as in the
paper.

In \autoref{sec:str}, we defined different semi-stubborn sets, i.e.,
$\sst^\shortleftarrow_\sigma(B)$,  $\sst^\leftrightarrow_\sigma(B)$ and  $\sst^\shortrightarrow_\sigma(B)$.
%The first, $\sst^\shortleftarrow_\sigma(B)$, is similar to stubborn sets defined in~\cite{guardpor2}.
%Since \cite{guardpor2} does not provide a proof and 
%the definition differs from the weak stubborn sets in~\cite{valmari1},
We first provide a proof or \autoref{th:sst}.
%We base the proof on a more general version of the theorem that also
%supports action non-determinism for semi-$\rightarrow$-stubborn sets,  $\sst^\shortrightarrow_\sigma(B)$.
%The version makes use of one additional,
%optional requirement \textbf{\ref{i:d1p}}. Let $c_\alpha \defn \set{\sigma\mid \tuple{\sigma,\alpha,\sigma'}\in \trans}$, now:
%
%\begin{provisosp}
%%  \setcounter{provisosi}{0}
%\item \label{i:d1p}\vspace{-1mm}
%			$\forall \alpha\in B \cap \en(\sigma),
%%			\sigma_1\tr{\alpha}{} \sigma_2
%			\beta \colon 
%			c_\alpha \lrestr \tr{\beta}{} \rrestr \overline{c_\alpha} \neq \emptyset \land
%			\overline{c_\alpha} \lrestr \tr{\beta}{} \rrestr c_\alpha \neq\emptyset \implies
%			\beta\in B$\\
%			(Actions $\beta$ both disabling some $\alpha\in B\cap\en(\sigma)$,
%			i.e., $c_\alpha \lrestr \tr{\beta}{} \rrestr \overline{c_\alpha} \neq \emptyset $,
%			and enabling $\alpha$, i.e., $\overline{c_\alpha} \lrestr \tr{\beta}{} \rrestr c_\alpha \neq\emptyset$
%			are included.)
%\end{provisosp}
%
%
%\begin{theorem}\label{th:sst}%stub sets for reduced state space generation
%If $B\subseteq \actions$, $\sst^{\star}_\sigma(B)$ and $\sigma \tr{\beta}{}\sigma'$
%for $\beta\notin B$,
%then $\sst^{\star}_{\sigma'}(B)$ for $\star \in \set{\leftarrow,\leftrightarrow}$,
%as well as for $\star \in \set{\rightarrow}$
%provided that $\beta$ does not disable
%a stubborn action, i.e., $\en(\sigma) \cap B \subseteq \en(\sigma') \cap B$,
%and $B$ is action-deterministic in $\sigma$, i.e., $\sizeof{B\cap \en(\sigma)}= 1$ or
%$\sst^\shortrightarrow_\sigma(B)$ also fulfills \textbf{\ref{i:d1p}}.
%\end{theorem}

\setcounter{theorem}{1}
%!TEX root = ../main.tex
\begin{theorem}\label{th:sst}%stub sets for reduced state space generation
If $B\subseteq \actions$, $\sst^{\star}_\sigma(B)$ and $\sigma \tr{\beta}{}\sigma'$
for $\beta\notin B$,
then $\sst^{\star}_{\sigma'}(B)$ for $\star \in \set{\shortleftarrow,\leftrightarrow}$,
as well as for $\star \in \set{\shortrightarrow}$
provided that $\beta$ does not disable
a stubborn action, i.e., $\en(\sigma) \cap B \subseteq \en(\sigma') \cap B$.
%,
%and $B$ is action-deterministic in $\sigma$, i.e., $\sizeof{B\cap \en(\sigma)}= 1$.
\end{theorem}

\begin{proof}
%The first part of the proof holds for $\star \in \set{\leftarrow,\rightarrow,\leftrightarrow}$.

Let $B$, $\beta$, $\sigma$ and $\sigma'$ be such that they satisfy
the premise of the theorem and $\alpha\in B$. We distinguish two cases:

If $\alpha\in \dis(\sigma)$, then 
let $\alpha,E$ be such that $E\in \nes_\sigma(\alpha)$ and $E \subseteq B$.
\textbf{\ref{i:d2}} remains valid for it in~$\sigma'$, since
$\beta$ cannot enable $\alpha$ because \textbf{\ref{i:d2}} holds in $\sigma$
and, by definition of NESs, we have that $E\in\nes_{\sigma'}(\alpha)$.

If $\alpha\in \en(\sigma)$, then either
$\alpha\in \en(\sigma')$ or $\alpha\in \dis(\sigma')$.
In the former case, the conclusion of the theorem is satisfied trivially,
as \textbf{\ref{i:d1}} also holds in $\sigma'$.
For the latter case, i.e. $\alpha\in \dis(\sigma')$, we consider
each $\star \in \set{\leftarrow,\rightarrow,\leftrightarrow}$ separately.

\begin{description}[labelsep=8pt,labelindent=0\parindent,itemindent=10pt,leftmargin=0pt,listparindent=4em]
\item[$\star = \leftrightarrow$:]
The proof is concluded, as the definition of strong commutativity $\comm$,
e.g., as the deterministic case illustrated by \autoref{eq:strong},
ensures that if $\beta$ disables $\alpha$, then the conclusion is not met.
(Note that this also concludes the proof of \autoref{th:valmari}.)

\item[$\star = \rightarrow$:]
The proof is concluded, because the additional `provided' condition
that $\en(\sigma) \cap B \subseteq \en(\sigma') \cap B$ ensures
 that $\beta$ cannot disable $\alpha$.

\item[$\star = \leftarrow$:]
From \textbf{\ref{i:d1}}, we have $\alpha \lcomm \gamma$ for all
$\gamma \notin B$. %However, $\gamma$ might disable $\alpha$.
Since we also have $\gamma \rcomm \alpha$, no $\gamma\notin B$ ever (re-)enables $\alpha$
by definition of right commutativity, as discussed in \autoref{sec:prelim}. 
%Since also in $\sigma'$ we have that $\gamma \notin B \colon \gamma \rcomm \alpha$,
Therefore, \textbf{\ref{i:d2}} holds in $\sigma'$
(there must be some $E\in \nes_{\sigma'}(\alpha)$ such that $E \cap \overline B = \emptyset$, hence $E\subseteq B$),
yielding again the conclusion of the theorem.
\end{description}
These three cases conclude the proof. \qed
\end{proof}

Before proving the monotonicity lemmata,
we recall the definition of dynamic movers and \autoref{def:trs}:
\setcounter{equation}{2}
%!TEX root = ../main.tex
\begin{IEEEeqnarray}{lll}\label{def:m}
\rem\rem\rem\rem\rem\rem\rem\rem\rem\rem\rem
\lmv_i(\sigma) &\defn %\{ \sigma\in\states \mid&
				 &\exists B \colon
		\sst_{\sigma}^\shortleftarrow(B),~
					B \cap \en(\sigma) = \actions_i \cap \en(\sigma)
		\label{eq:left}\\*
\rem\rem\rem\rem\rem\rem\rem\rem\rem\rem\rem
\rmv_i({\sigma,\alpha,\sigma'})
			&\defn %\{ \tuple{\sigma,\alpha,\sigma'}\rem \in T_i \mid
			&\exists B \colon
		\sst_{\sigma}^\shortrightarrow(B),~\alpha\in B,~
		 B\cap \en(\sigma') \subseteq \actions_i,~
		 B\cap \en(\sigma) = \set{\alpha}%,~
%		 \en(\sigma)\rem\setminus\rem B \subseteq  \en(\sigma')\rem\setminus\rem B~
		  %\hspace{1em}(, \tuple{\sigma,\alpha,\sigma'}\in\trans_i)
	\label{eq:exclude}
\vspace{-1ex}
\end{IEEEeqnarray}

\setcounter{lemma}{0}
%!TEX root = ../main.tex
\begin{lemma}\label{lem:dlm}
The dynamic left-moving property is never
remotely disabled, i.e.:
if $\lmv_i(\sigma_1) \land i\neq j \land \sigma_1\tr{\beta}{j} \sigma_2$,
then $\lmv_i(\sigma_2)$.
\end{lemma}
\begin{proof}
Assume the premise:
$\lmv_i(\sigma_1)$ with $i\neq j$ and $\sigma_1\tr{\beta}{j} \sigma_2$. We derive the conclusion.

Let $B$ be such that
$\sst_{\sigma_1}^\shortleftarrow(B)$ and
$B\cap \en(\sigma_1) = \actions_i \cap \en(\sigma_1)$.
As $j\neq i$, we may apply \autoref{th:sst} to find that
$B$ is also a valid semi-$\shortleftarrow$-stubborn set
in~$\sigma_2$, i.e. $\sst_{\sigma_2}^\shortleftarrow(B)$.
Moreover, $\beta$ cannot enable any $\gamma\in B\cap \dis(\sigma_1)$ by \textbf{\ref{i:d1}},
hence $B\cap \en(\sigma_2) = \actions_i \cap \en(\sigma_2)$.
That together with the semi-$\shortleftarrow$-stubbornness of $B$ in $\sigma_2$, implies that $\lmv_i(\sigma_2)$.
\qed
\end{proof}

%!TEX root = ../main.tex
\begin{lemma}\label{lem:drm}
Dynamic right-movers retain dynamic moveability after moving, i.e.:
if $\rmv_i(\sigma_1,\alpha,\sigma_2) \land \sigma_1\tr{\alpha}{i} \sigma_2 \tr{\beta}{j} \sigma_3$ for $i\neq j $,
then $\exists \sigma_1 \tr{\beta}{j}\sigma_4 \colon \rmv_i(\sigma_4,\alpha,\sigma_3)$.
\end{lemma}
\begin{proof}
Assume the premise: $\rmv_{i}(\sigma_1,\alpha,\sigma_2)$ for $\alpha\in\actions_i$
and $\sigma_1\tr{\alpha}{i} \sigma_2 \tr{\beta}{j} \sigma_3$ for $j\neq i$.
Let $B \subseteq \actions$ satisfy \autoref{eq:exclude}, i.e.:
$\sst_{\sigma_1}^\shortrightarrow(B),~\alpha\in B$, 
$B\cap \en(\sigma_2) \subseteq \actions_i$, and
$B\cap \en(\sigma_1) = \set{\alpha}$.
We derive the conclusion, i.e.: $\exists \sigma_1 \tr{\beta}{j}\sigma_4$ such that
$\sst_{\sigma_4}^\shortrightarrow(B),~\alpha\in B$, 
$B\cap \en(\sigma_3) \subseteq \actions_i$, and
$B\cap \en(\sigma_4) = \set{\alpha}$.

First we show that $\exists \sigma_1 \tr{\beta}{j}\sigma_4$ and
 $B \cap \en(\sigma_4) = \set{\alpha}$.
As $\beta\in\en(\sigma_2)$, we obtain $\beta\notin B$ (since $\beta\notin \actions_i$).
Since therefore $\alpha$ right-commutes with $\beta$ by the contraposition of
\textbf{\ref{i:d1}}, we obtain the commuting path $\sigma_1\tr{\beta}j \sigma_4\tr{\alpha}i \sigma_3$.
Assume $\exists \gamma \in B \cap \en(\sigma_4) \setminus \set{\alpha}$.
Action $\beta$ must have enabled $\gamma$, otherwise $\gamma \in \en(\sigma_1)$,
contradicting our assumption that $B\cap \en(\sigma_1) = \set\alpha$.
Now, if $\sigma_1\tr{\beta}j \sigma_4$ enables~$\gamma$,
by \textbf{\ref{i:d2}}, also $\beta\in B$, again contradicting the assumption.
Therefore, we have  $B \cap \en(\sigma_4) = \set{\alpha}$.

\autoref{th:sst} tells us that $B$ with $\sst_{\sigma_1}^\shortrightarrow(B)$
is also semi-stubborn in $\sigma_4$, i.e. $\sst^\shortrightarrow_{\sigma_4}(B)$
(\autoref{th:sst}'s additional condition that
$\en(\sigma_1) \cap B \subseteq \en(\sigma_4) \cap B$
is met because $\en(\sigma_1) \cap B = \en(\sigma_4) \cap B = \set\alpha$ 
%is enabled
%and any non-deterministic $\alpha'$ are enabled 
%in $\sigma_1$, and it is also enabled in $\sigma_4$
 as shown above).

We now show that $B \cap \en(\sigma_3) \subseteq \actions_i$ also holds.
Assume $\exists \gamma \in B \cap \en(\sigma_3) \setminus \actions_i$.
Action $\beta$ must have enabled $\gamma$, otherwise $\gamma \in \en(\sigma_2)$,
contradicting our assumption that $B\cap \en(\sigma_2) \subseteq \actions_i$.
However, if $\sigma_2\tr{\beta}j \sigma_3$ enables $\gamma$,
by \textbf{\ref{i:d2}}, also $\beta\in B$, again contradicting our assumptions.
Therefore, we have $B \cap \en(\sigma_3) \subseteq \actions_i$.

The above shows that
%
%As $\beta$ cannot enable any action in $B$ by \textbf{\ref{i:d2}},
%we finally also get 
%$\en(\sigma_4)\rem\setminus\rem B \subseteq  \en(\sigma_3)\rem\setminus\rem B$ from
%$\en(\sigma_4)\rem\setminus\rem B \subseteq  \en(\sigma_3)\rem\setminus\rem B$.
%Hence, we can conclude that
$\rmv_i(\sigma_4,\alpha,\sigma_3) $.
\qed
\end{proof}

% TODO: generalize dyn right movers?
%At this point it is interesting to note that we can also prove \autoref{lem:drm}
%for a dynamic right mover definition that allows multiple 
%mutually non-deterministic actions from $\actions_i$ to be included in the stubborn set,
%shown as $N\rmv_i$ below. 
%\[
%N\rmv_i({\sigma,\alpha,\sigma'})
%			&\defn %\{ \tuple{\sigma,\alpha,\sigma'}\rem \in T_i \mid
%			&\exists B \colon
%		\sst_{\sigma}^\rightarrow(B),~\alpha\in B,~
%		 B\cap \en(\sigma') \subseteq \actions_i%,~
%\]
%\begin{provisosp}
%%  \setcounter{provisosi}{0}
%\item \label{i:d1p}\vspace{-1mm}
%			$\forall \alpha\in B \cap \en(\sigma),
%%			\sigma_1\tr{\alpha}{} \sigma_2
%			\beta \colon 
%			c_\alpha \lrestr \tr{\beta}{} \rrestr \overline{c_\alpha} \neq \emptyset \land
%			\overline{c_\alpha} \lrestr \tr{\beta}{} \rrestr c_\alpha \neq\emptyset \implies
%			\beta\in B$\\
%			(Actions $\beta$ both disabling some $\alpha\in B\cap\en(\sigma)$,
%			i.e., $c_\alpha \lrestr \tr{\beta}{} \rrestr \overline{c_\alpha} \neq \emptyset $,
%			and enabling $\alpha$, i.e., $\overline{c_\alpha} \lrestr \tr{\beta}{} \rrestr c_\alpha \neq\emptyset$
%			are included.)
%\end{provisosp}

Recalling \autoref{def:trs}, we see that proving its preservation property is easy:

\setcounter{equation}{4}
\setcounter{definition}{1}
\begin{definition}[Transaction system]\label{def:trs}
Let $H \defn \set{\NN, \RR, \LL} ^ P$ be an array of local phases.
The transaction system is CTS $\cts'\defn \tuple{\tstates, \ttrans, \actions, \tsint}$ such that:
\vspace{-2ex}
%!TEX root = ../main.tex
\begin{IEEEeqnarray}{rllrl}
\hspace{-5ex}\tstates 	\defn &\parbox{0cm}{\mbox{$\states \times \phases,~~~~~~~~~~
					\tsint 	\defn\tuple{\sint, \NN^{{\procs}}}	$}}				\nonumber\\*
%\hspace{-5ex}\tsint 	\defn &\parbox{0cm}{\mbox{$\tuple{\sint, \NN^{{\procs}}}	$}}		\nonumber\\*
\hspace{-5ex}\ttrans_i 	\defn &\parbox{0cm}{\mbox{$\{ \tuple{\tuple{\sigma,h},\alpha,\tuple{\sigma',h'}} \in{\tstates}\times\actions\times\tstates \mid
				(\sigma,\alpha,\sigma')\in T_i
				,\,\,\, \forall j\neq i\colon h_j'=h_j,$}}			\nonumber\\*
%\hspace{-5ex} &\parbox{0cm}{\mbox{~~~~~~~~~~~~~~~~~~~~~~~~~
%				$\forall j\colon \phase_j = \LL \implies \lmv_j(\sigma)\land \en(\sigma)\cap\actions_j\neq \emptyset,$}}\\*
&& \RR &\hphantom{|XX}& \text{iff } h_i \neq \LL
						\,\,\,\land\,\,\,
						\rmv_i(\sigma,\alpha,\sigma')
						\,\,\,\land\,\,\,
						\alpha\notin \actions_\ominus^\varphi			\label{eq:1} \\*
&\smash{h_i'\hspace{-.5ex}=\hspace{-.5ex}\left\{\IEEEstrut[8\jot] \right.}
& \LL && \text{if } \lmv_i(\sigma')
						\,\,\,\land\,\,\, 
						\en(\sigma')\cap\actions_i\cap \actions_\oplus^\varphi=\emptyset	\label{eq:2}\\*
&& \NN && \text{otherwise (or as alternative when \autoref{eq:2} holds)} 				\label{eq:3}\\*
&  \parbox{0cm}{\mbox{$\}$%, with  %$\rmv_i\defn \bigcup_{\alpha\in\actions_i} \rmv_\alpha$.}}
%$\rmv_i\defn \set{(\sigma,\sigma')\mid \bigcup_{\alpha\in\actions_i}(\sigma,\alpha,\sigma')\in\rmv_\alpha}$.
}}
%\text{, and~}\ttrans \defn \biguplus_i \ttrans_i
%,\text{ where } \LL_i \defn \set{ \tuple{\sigma,\phase} \mid \phase_i = \LL }$}}
\nonumber
\vspace{-2ex}
\end{IEEEeqnarray}

\noindent
%$\textsf{X}_i \defn \set{ \tuple{\sigma,\phase} \mid \phase_i = \textsf{X} }$
%for $\textsf{X} \in \set{\LL, \RR, \NN}$.

\end{definition}

\begin{lemma}\label{lem:preserves}
\autoref{def:trs} preserves invariants:
$\reach(\cts) \models\Box\varphi \Leftrightarrow\reach(\cts')\models \Box\varphi$.
\end{lemma}
\begin{proof}
The definition ensures the bisimulation: $\set{\tuple{\sigma, \tuple{\sigma,h}}\in \states\times \tstates}$.
\qed
\end{proof}

Towards proving \autoref{th:reduction}, we first recall the main theorem from
\cite{vmcai}.
\autoref{th:vmcai} requires one
bisimulation $\cong_i$ for each thread $i$ %(see \autoref{eq:bisim})
and a weakened definition of commutativity \concept{up to bisimulation}.
We recall these definitions first from~\cite{vmcai}.

We now formally define the notion of \concept{thread bisimulation} required for the reduction, as well as commutativity up to bisimilarity.

%\begin{definition}[Thread-bisimulation]
%A relation $R$ is a thread-bisimulation on a TS if and only if 
%$$\sigma R \sigma' \wedge \sigma\hspace{-1mm} \to_k\hspace{-1mm} \sigma_1 \Rightarrow \exists \sigma_1'\colon \sigma'\hspace{-1mm} \to_k\hspace{-1mm} \sigma_1' \wedge \sigma_1 \cong_i \sigma_1'.$$
%For simplicity, we will often just say bisimulation for thread-bisimulation. 
%\end{definition}

\begin{definition}[thread bisimulation]\label{def:bisim}
An equivalence relation $R$ on the states of a CTS $\tuple{\states, T, \actions, \sigma_0}$
is a thread bisimulation iff 

\centering
\begin{tikzpicture}

   \tikzstyle{e}=[minimum width=1cm]
   \tikzstyle{every node}=[font=\small, node distance=1.2cm]

  \node (s1) {$\sigma$};
  \node (s2) [below of=s1] {$\sigma'$};
  \node (s3) [right of=s1] {$\sigma_1$};

  \path (s1) edge[-] node[pos=.5,sloped,above] (nn) {$R$} (s2);
  \path (s1) -- node(m)[midway,sloped]{$\to_i$} (s3);

  % grey stuff

  \node (s1p) [gray,right of=s3,xshift=1cm] {$\sigma$};
  \node (s2p) [below of=s1p] {$\sigma'$};
  \node (s3p) [right of=s1p] {$\sigma_1$};

  \path (s1p) edge[-] node[gray,pos=.5,sloped,above] (nnn) {$R$} (s2p);
  \path (s1p) -- node(m)[gray,midway,sloped]{$\to_i$} (s3p);

  \node (s4p) [right of=s2p] {$\sigma'_1$};
  \path (s2p) -- node [midway,sloped]{$\to_i$} (s4p);
  \path (s3p) edge[-] node[pos=.5,sloped,above]{$R$} (s4p);

  \node (n) [left of=nnn] {$\implies \exists \sigma_1'\colon$};

  \node (n) [left of=nn] {$\forall \sigma,\sigma',\sigma_1,i\colon$};
  
\end{tikzpicture}
\end{definition}

Standard bisimulation is an equivalence relation $R$ which satisfies the property 
from \autoref{def:bisim} when the indexes $i$ of the transitions are removed. 
Hence, in a thread bisimulation, in contrast to standard bisimulation,
the transitions performed by  thread $i$ will be matched
by transitions performed by the same thread~$i$.
As we only make use of thread bisimulations, we will often refer to them simply as
bisimulations.

%\begin{align}\label{eq:bisim}
%\cong_i\hspace{-1mm}\circ \trp \,\,\subseteq\,\, \trp\circ \cong_i \,\,\land \,\,
%\cong_i^{-1}\hspace{-1mm}\circ \trp \,\,\subseteq\,\, \trp\circ \cong_i^{-1}
%\end{align}

We can lift these bisimulations to sets of threads, by taking the
equivalence closure, e.g.
$\cong_Z$ being the transitive closure of the union of all $\cong_i$
for $i \in Z$. Note that $\cong_i \Leftrightarrow \cong_{\set i}$.
With this we can also refine commutativity as follows.

\begin{definition}[commutativity up to bisimulation] \label{def:comm-bisim}
Let $R$ be a thread bisimulation on a CTS $\tuple{\states, T, \actions, \sigma_0}$.
The right and left commutativity up to $R$ of the transition relation $\to_i$
with $\to_j$, notation
$\to_i\, \,\,\rcomm_{R}\,\, \to_j$ /$\to_i\, \,\,\lcomm_{R}\,\, \to_j$ are defined as follows.
\begin{IEEEeqnarray}{lCllr}
\tr{}{i}\,\,\rcomm_R\,\, \tr{}{j} &\defn &\tr{}{i} \circ \tr{}{j} \circ R
	&\,\,\subseteq\,\, \tr{}{j} \circ \tr{}{i} \circ R\phantom{XXX}
	&\text{~($\rcomm$ up to $R$)}\nonumber\\*
\tr{}{i}\,\,\lcomm_R\,\, \tr{}{j} &\defn &\tr{}{i} \circ \tr{}{j} \circ R
	&\,\,\supseteq\,\, \tr{}{j} \circ \tr{}{i} \circ R
	&\text{~($\lcomm$ up to $R$)}\nonumber
\end{IEEEeqnarray}
\noindent
Illustratively: % (replacing actions $\alpha_1$ and $\alpha_2$ with threads $i$ and $j$):

\noindent
$\to_i\, \rcomm_{R} \to_j\,\,\,\, \Longleftrightarrow$\hspace{10em}
$\to_i\, \lcomm_{R} \to_j\,\,\,\, \Longleftrightarrow$\\
\begin{tikzpicture}
   \tikzstyle{e}=[minimum width=1cm]
   \tikzstyle{every node}=[font=\small, node distance=.45cm]

  \node (s1) {$\sigma_1$};
  \node (s2) [node distance=1.2cm,below of=s1] {$\sigma_2$};
  \node (s3) [right of=s2, xshift=.6cm] {$\sigma_3$};
  \path (s2) -- node[pos=.45]{$\to_j$} (s3);
  \path (s1) -- node(m)[midway,sloped]{$\to_i$} (s2);

  \node (s1n) [xshift=2.5cm,right of=s1] {$\sigma_1$};
  % grey stuff
  \node (s2n) [gray,node distance=1.2cm,below of=s1n] {$\sigma_2$};
  \node (s3n) [gray,right of=s2n, xshift=.6cm] {$\sigma_3$};
  \path (s2n) -- node[gray,pos=.45]{$\to_j$} (s3n);
  \path (s1n) -- node(m)[gray,midway,sloped]{$\to_i$} (s2n);

  \node (s0n) [left of=m, xshift=-.4cm]  {$\implies \exists \sigma_3',\sigma_4\colon$};

  \node (s4n) [xshift=.9cm,right of=s1n] {$\sigma_4$};
  \path (s1n) -- node [midway,sloped]{$\to_j$} (s4n);

  \node (s3pn) [node distance=1.1cm,right of=s3n] {$\sigma_3'$};
  \path (s4n) -- node[midway,sloped]{$\to_i$} (s3pn);

  % sets
  \path (s3pn) -- node(AA)[sloped,pos=.1]{} (s3n);

  \node (s3pn) [node distance=1em,below of=AA,xshift=-.3cm] {$\tuple{\sigma_3,\sigma_3'}\in R$};
  
\end{tikzpicture}
~~%\qquad
%\raisebox{1.6cm}{
\begin{tikzpicture}

   \tikzstyle{e}=[minimum width=1cm]
   \tikzstyle{every node}=[font=\small, node distance=.45cm]

  \node (s1) {$\sigma_1$};
  \node (s2) [right of=s1, xshift=.6cm] {$\sigma_2$};
  \node (s3) [node distance=1.2cm,below of=s2, xshift=.9cm] {$\sigma_3$};
  \path (s2) -- node(m)[midway,sloped]{$\to_j$} (s3);
  \path (s1) -- node[midway,sloped]{$\to_i$} (s2);

  \node (s1n) [xshift=2.2cm,right of=s2] {$\sigma_1$};
  % grey stuff
  \node (s2n) [gray,right of=s1n, xshift=.6cm] {$\sigma_2$};
  \node (s3n) [gray,node distance=1.2cm,below of=s2n, xshift=.9cm] {$\sigma_3$};
  \path (s2n) -- node[gray,midway,sloped]{$\to_j$} (s3n);
  \path (s1n) -- node[gray,midway,sloped]{$\to_i$} (s2n);

  \node (s4n) [node distance=1.2cm,below of=s1n] {$\sigma_4$};
  \path (s1n) -- node(m) [midway,sloped]{$\to_j$} (s4n);

  \node (s3pn) [node distance=.9cm,right of=s4n] {$\sigma_3'$};
  \path (s4n) -- node[midway,sloped]{$\to_i$} (s3pn);

  \node (s0n) [left of=m, xshift=-.4cm]  {$\implies \exists \sigma_3',\sigma_4\colon$};
  
  % sets
  \path (s3pn) -- node(AA)[sloped,pos=.1]{} (s3n);

  \node (s3pn) [node distance=1em,below of=AA] {$\tuple{\sigma_3,\sigma_3'}\in R$};
  
\end{tikzpicture}
%}\\
%\begin{tikzpicture}
%   \tikzstyle{e}=[minimum width=1cm]
%   \tikzstyle{every node}=[font=\small, node distance=.45cm]
%
%  \node (s1) {$\sigma_1$};
%  \node (s2) [right of=s1, xshift=.9cm] {$\sigma_2$};
%  \path (s1) -- node[pos=.45]{$\to_j$} (s2);
%
%
%  % grey stuff
%  \node (s4) [node distance=1.2cm,gray,below of=s1] {$\sigma_4$};
%  \path (s1) -- node(m)[gray,midway,sloped]{$\to_i$} (s4);
%  \node (s3p) [gray,xshift=.6cm,right of=s4] {$\sigma_3'$};
%  \node (s3) [node distance=1.1cm,right of=s3p] {$\sigma_3$};
%  \path (s4) -- node[gray,pos=.48]{$\to_j$} (s3p);
%
%  \path (s2) -- node[midway,sloped]{$\to_i$} (s3);
%  \node (s0) [left of=m, xshift=-.6cm]  {$\to_i \lcomm_{R} \to_j \defn $};
%  % sets
%  \path (s3p) -- node[sloped,pos=.48]{$\cong_{R}$} (s3);
%
%  \path (s2) -- node[sloped,pos=.48]{$\Leftarrow$} (s4);
%\end{tikzpicture}
\end{definition}

%\begin{align}\label{eq:rcommupto}
%\text{\input{figures/dyn-right-comm}}
%\end{align}

We write $\comm_Z$ for $\comm_{\cong_Z}$.

Using these definitions, \autoref{th:vmcai} provides an axiomatization
of the properties required for reducing the CTS using dynamic TR.
The theorem is similar to the reduction theorem in~\cite{vmcai},
where it is explained in detail. A proof of correctness is
provided in \cite{arxiv}.\footnote{The version in~\cite{arxiv}
does not include \autoref{i:vispre} and \autoref{i:vispost}.
To reason over invariant violations,
it instead distinguishes separate error states $\textsf{Err}_i\subseteq \NN_i$.
Using the path provided by \cite[Th.~2]{arxiv}, it is straightforward to show
that if a bad state $\overline\varphi$ is reachable in the complete system,
then so is one reachable in the reduced system. See also the
explanation of \textbf{L4} at the end of \autoref{sec:prelim}. 
}
Most of the constraints in its premise mirror the constraints
\textbf{L1--L4} provided in \autoref{sec:prelim}.
The commutativity condition
however is weakened to allow commutativity up to
bisimulation. Further conditions constrain the
phases of the transaction system with respect to
the newly added thread bisimulations.

\setcounter{theorem}{3}
{
\def\NN{N}
\def\LL{L}
\def\RR{R}
\def\varphi{Y}
\def\states{X}
\begin{theorem}[Reduction]
\label{th:vmcai}
\noindent
Let $\tuple{\states, T, \actions, \sigma_0}$
be a concurrent transition system, $\varphi\subseteq \states$
and $\to_i \defn \set{\tuple{\sigma,\sigma'} \mid \tuple{\sigma,\alpha,\sigma'} \in T_i}$
(as usual).
For each thread $i$, there exists a thread bisimulation relation $\cong_i$.
For all $i, j\neq i$ the following holds:

\begin{enumerate}
\setcounter{enumi}{0}
%\item $\sizeof\states < \infty$
%\hfill the system is finite
%\item $\Ta = \biguplus_i \to_i$ \label{i:threads}
%\hfill threads have disjoint actions
%\item $W_i \subseteq N_i$, \label{i:error}
%\hfill local errors are external

\item $\states = \RR_i\uplus \LL_i\uplus \NN_i$, \label{i:part}
\hfill  ($\RR_i,\LL_i, \NN_i$ (Pre, post and external) partition $\states$)
\item $\to_i \subseteq  \RR^2_j\cup \LL^2_j\cup \NN^2_j$
\hfill ($\to_i$ is invariant over partitions of~$j$\label{i:invar})

\item $\LL_i \lrestr \to_i \rrestr \RR_i  = \emptyset$\label{i:post}
\hfill (post does not locally reach pre)

\item $\to_i \rrestr \RR_i  \rcomm_{\set{j}} \to_j $\label{i:right}
\hfill ($\to_i$ ending in pre right commutes with $\to_j$)

\item $ \LL_i \lrestr \to_i \,\,\lcomm_{\set{i,j}} \to_j $ \label{i:left}
\hfill ($\to_i$ starting from post left commutes with $\to_j$)

\item $\forall \sigma\in L_i\colon\exists \sigma'\in N_i\colon \sigma\to_i^{*}\sigma'$
\label{i:fairness}
\hfill (post phases terminate locally)

\item\label{i:bisimdisjoint}
	$\cong_i\,\subseteq {\LL}_j^2\cup{\RR}_j^2\cup{\NN}_j^2 \vphantom{\overline{\NN_i}^2}$
	\hfill
	($\cong_i$ entails $j$-phase-equality)

\item $\varphi\lrestr (\to_i \rrestr \RR_i) \rrestr \overline\varphi = \emptyset$
\label{i:vispre}
\hfill ($\to_i$ into pre does not disable $\varphi$)
\item $\overline\varphi \lrestr(\LL_i \lrestr \to_i) \rrestr \varphi = \emptyset$
\label{i:vispost}
\hfill ($\to_i$ from post does not enable $\varphi$)
\end{enumerate}

\noindent
Let 
%$W \defn \bigcup_i W_i$, $N \defn \bigcap_i N_i$ and 
%$N_\tx \defn \bigcap_{u\neq t} N_j$.\\
$\trtrans_i \defn \bigcup_{j\neq i} \NN_j \lrestr \to_i$
\hfill
($i$ 
only transits when all $j$ are external).

\noindent
Let
$ \brtrans_i \defn \NN_i\lrestr (\trtrans_i\rrestr \overline{\NN_i})^* \trtrans_i \rrestr \NN_i$
\hfill
(skip internal states).

\noindent
Let  $\brtrans \defn \bigcup_i \brtrans_i$ and  $\NN\defn \bigcup_i \NN_i$.

\noindent
Now, if $\sigma\to^{*} \sigma'$ with $\sigma \in\NN\cap \varphi$ and $\sigma'\in \overline\varphi$, then
$\exists\sigma''\in \overline\varphi$ s.t. $\sigma \,\brtrans^{*} \,\sigma''$. 
\end{theorem}
}

We will show that
our transaction system of \autoref{def:trs} satisfies the premise of
\autoref{th:vmcai} (in the following \autoref{lem:th}).
In the process, the most important aspects of the
theorem, i.e. the movability up to bisimulation $\cong_X$ 
in \autoref{i:right} and \autoref{i:left}, is explained.
Notice that the in the right mover case,
we have $X = \set{j}$, while in the left-mover cases we have $X=\set{i,j}$.

To see the challenge ahead, observe that a remote thread $j$ can 
activate a dynamic mover of thread $i$.
We illustrate with an example that this dynamic behavior causes
loss of commutativity in the \emph{transaction} system
(not in the underlying \emph{transition} system),
because of the phase information that the transaction system tracks.
In the following, let $q_x\defn \tuple{\sigma_x,\phase_x}$ and
$q_x' \defn \tuple{\sigma_x',\phase_x'}$ for $x\in \nat$, so that we can
easily track related states in both systems.

%Consider an example program \ccode{lock(m); $\ldots \|$ unlock(m); $\ldots$} and 
%let $\alpha$ be the lock action in thread $i$ and $\beta$
%the unlock action in thread $j$ of this example.
%Because $\alpha$ can only be executed once,
%$\beta$ dynamically becomes both-mover after $\alpha$ is executed (and vice versa).\footnote{The
%example seems artificial because the lock and unlock reside in
%the different threads.
%Nonetheless, it represents various other scenarios
%in which right or left moveability can be activiated, e.g.
%when the two threads synchronize on locks with a third and fourth thread
%but also access the same shared variables in the same action.}

Let %$\sigma_1\in c_\alpha^{\to}$ and
$\tuple{\sigma_1,\alpha,\sigma_2} \in T_\alpha$
(in the \emph{transition} system).
We have $\tuple{\sigma_2,\beta,\sigma_3} \in \rmv_j$, i.e.
$\beta$ is dynamic right moving (in $\sigma_2$) %, enabled from $\sigma_2$
and leads to $\sigma_3$.
Because of its movability, the \emph{transaction} system allows that $q_3 \in \RR_i$
(see \autoref{eq:2} of \autoref{def:trs}) as the following figure shows.
The figure shows the right move of $\alpha$ (also a right mover) with respect to $\beta$ (the gray part).
The yields states $q_3'$ and $q_4$ where $\beta$ is executed before $\alpha$.
{
\newcommand\sigmaold{\sigma}
\def\sigma{q}
\begin{center}
%!TEX root = ../main.tex

%RIGHT MOVERS

%\def\to{\rightarrow}

\begin{tikzpicture}[baseline={([yshift=-.5ex]current bounding box.center)}]
   \tikzstyle{e}=[minimum width=1cm]
   \tikzstyle{every node}=[font=\small, node distance=.5cm]

  \node (s1) {$\sigma_1$};
  \node (s2) [node distance=1.2cm,below of=s1] {$\sigma_2$};
  \node (s3) [right of=s2, xshift=.5cm] {$\sigma_3$};
  \path (s2.east) -- node[pos=.45]{$\stackrel\beta\to_j$} (s3.west);
  \path (s1.south) -- node[midway,sloped]{$\stackrel\alpha\to_i$} (s2.north);

  % grey stuff
  \node (s4) [xshift=.9cm,gray,right of=s1] {$\sigma_4$};
  \node (s3p) [gray,node distance=1.2cm,below right of=s4,xshift=-.5cm] {$\sigma_3'$};
  \path (s1.east) -- node[gray,midway,sloped]{$\stackrel\beta\to_j$}
  (s4.west);
  \path (s4.south) -- node[gray,midway,sloped]{$\stackrel\alpha\to_i$} (s3p);

  % sets
  \node (S1) [left of=s1, xshift=-.2cm,e] {$\RR_j\ni$};
  %\node (S2) [left of=s2, xshift=-.3cm,e] {$\RR_j \ni$}; %\RR_i
  \node (S2) [below right of=s3,xshift=.3cm, yshift=.2cm,e] {$\in\RR_j$}; %\RR_i

  \node (S4) [right of=s4, gray,xshift=.27cm,e] {$\in\LL_j$};
  \node (S3) [left of=s3, xshift=-1.2cm,e] {$ \RR_j \ni$}; % \RR_i

  \path (s3) -- node[gray,sloped,pos=.48]{$\neq$} (s3p);

 \node (S3) [gray,right of=s3p, xshift=.27cm,e] {$\in \LL_j$};  
%  
%  \node (s5) [node distance=1.2cm,below right of=s3] {$\sigma_5$};
%  \path (s3) -- node[sloped,pos=.48]{$\to_j$} (s5);
%  \node (S3) [left of=s5, xshift=-.07cm,e] {$\RR_j\ni$};
%  
%  \node (s5p) [gray,node distance=1.2cm,below right of=s3p] {$\sigma_5'$};
%  \path (s5) -- node[gray,sloped,pos=.48]{$\cong_j$} (s5p);
%  \path (s3p) -- node[gray,sloped,pos=.48]{$\to_j$} (s5p);
%  \node (S3) [gray,right of=s5p, xshift=.07cm,e] {$\in \LL_j$};  
%  
%  
\end{tikzpicture}
\end{center}
}
%Now assume that it occurs that
%$\tuple{\sigma_1,\beta,\sigma_4}\notin M_\beta^{\rightarrow}$.
%Indeed it could be that 
%$\alpha \notin \fst^\rightarrow_{\sigma_2}(B)$ for some set $B$,
%but $\nexists B' \colon \alpha \in \fst^\rightarrow_{\sigma_1}(B')$.
%This could happen when $\alpha$ is a lock and $\beta$ an unlock operation,
%i.e. a strict left-mover in the presence of locks.
%If $\alpha$ cannot be executed again from
%$\sigma_2$ then $\beta$ dynamically becomes both-mover in $\sigma_2$.
Because e.g. whatever action $\gamma$ that does not commute with
$\beta$ became unreachable after $\alpha$,
$\beta$ does not right-move from $\sigma_1$ where $\alpha$ is not yet taken.
{
\newcommand\sigmaold{\sigma}
\def\sigma{q}
We see therefore that $\sigma_3'\in\LL_j$.
Additionally, we have  $\sigma_4\in\LL_j$ by \autoref{def:trs}
(see $\forall j\neq i\colon h_j'=h_j$).
Therefore, the moving operation does not commute in
the transaction system as $\sigma_3 \neq \sigma_3'$.

Our theorem accounts for the differing phases of
$\sigma_3$ and $\sigma_3'$.
}
(Apart from the $j$-phases, these states are indeed equivalent, i.e.:
%if $q_3 = \tuple{\sigma_3,\phase_3}$ and $q_3' = \tuple{\sigma_3',\phase_3'}$, then 
$\sigma_3 = \sigma_3'$ and $\forall k\neq j \colon \phase_3[k] = \phase_3'[k]$.)
To this end, the bisimulations abstract from the phase changes, 
showing that the behavior of the  transaction system mimics that of
the original transition~system.

Bisimulations indeed arise naturally from the introduced phase flags:
All transitions of a thread $i$
in the transaction system are copies from transitions
in the original transition system that end in a state with a
different $i$ phase. Therefore, by discarding the phase information for $i$
we end up with a bisimulation for $i$ (see \autoref{eq:abstract}).
\begin{align}\label{eq:abstract}
\tuple{\sigma,\phase} \cong_i \tuple{\sigma',\phase'} \Longleftrightarrow
\sigma = \sigma' \land \forall j\neq i \colon \phase_j = \phase'_j
\end{align}

\begin{lemma}\label{lem:th}
The transaction system in Def.~\ref{def:trs} fulfills the premise of \autoref{th:reduction}
(\autoref{i:part}--\ref{i:vispost}) with 
$X = \states'$, $N_i = \NN_i$, $R_i = \RR_i$, $L_i = \LL_i$ for all threads $i$,
provided that post-phases terminate, i.e. dm{	\def\sigma{q}
$\forall \sigma \in \LL_i \colon \exists \sigma'\in\NN_i \colon \sigma\hookrightarrow^*_i \sigma'$},
and are actuated as well, i.e.
{	\def\sigma{q}
$\forall \sigma \in \LL_i \colon
\exists \sigma',\sigma''  \colon \sigma' \hookrightarrow_i \sigma'' \hookrightarrow^* \sigma$}.
\end{lemma}

\begin{proof}\label{proof:th}
We take the $\cong_i$ relation from \autoref{eq:abstract} and show that it
is a valid thread bisimulation for each thread $i$.
Then we focus out attention to
the nine items in the premise of \autoref{th:vmcai} and
show how the transaction system fulfills these conditions.
In the following, again let $q_x \defn \tuple{\sigma_x,\phase_x}$ and
$q_x' \defn \tuple{\sigma_x',\phase_x'}$ for $x\in \nat$.

To see that the relation $\cong_i$ from \autoref{eq:abstract} is a correct 
thread bisimulation for thread $i$
according to \autoref{def:bisim}, 
assume that $\tuple{\tuple{\sigma,h}, \tuple{\sigma',h'}}\in\,\, \cong_i$ and
$\tuple{\tuple{\sigma,h},\alpha,\tuple{\sigma'',h''}}\in T'_i$.
By definition, we have
	$\sigma' = \sigma$ and
  $\forall j\neq i \colon h_j = h_j' = h_j''$.
Therefore, by \autoref{def:trs}, we also have
$\tuple{\tuple{\sigma',h'}, \alpha, \tuple{\sigma'', h'''}} \in\ttrans$
for some $h''' \in H$ such that
  $\forall j\neq i \colon h_j = h_j' = h_j''= h_j'''$.
Finally, by definition, $\tuple{\tuple{\sigma'',h''}, \tuple{\sigma''',h'''}} \in \,\,\cong_i$,
concluding the proof that $\cong_i$ is a proper thread bisimulation.

Next, we consider how \autoref{def:trs} satisfies the items of the premise of \autoref{th:vmcai}:
\begin{description}
\item[\phantom{XX}\autoref{i:part}]
By definition of \phases, we have $\forall i\colon \NN_i \uplus \RR_i \uplus \LL_i$.

\item[\phantom{XX}\autoref{i:invar}]
$T_i'$ of the
transaction system ensures that remote phases remain invariant:
$\forall \tuple{\tuple{\sigma,h},\alpha,\tuple{\sigma',h'}}\in T'_i
\colon  \forall   j\neq i \colon h'_j = h_j$.
Therefore, we have:\\ $\forall i\colon  \to'_i \,\,\subseteq\,\,  \NN_i^2 \cup \RR_i^2 \cup \LL_i^2$.

\item[\phantom{XX}\autoref{i:post}]
Follows immediately from \autoref{eq:1} in \autoref{def:trs}.

\item[\phantom{XX}\autoref{i:right}]
%\cite[Th.~5]{arxiv} demonstrates that the following condition
%(equivalent to the second case in \cite[Def.~13]{arxiv}) 
%is a sufficient for dynamic \emph{right} movability up to $\cong_{\set j}$:
%$(c^\rightarrow_\alpha \lrestr \tr{\alpha}{i}) \rcomm  \tr\beta{j}$.
%We show it is met.

Assume that $q_1\stackrel{\alpha}{\to}{\rem\rem}_i'\,\, q_2$ with $q_2 \in \RR_i$ and
			$q_2\stackrel{\beta}{\to}{\rem\rem}_j'\,\, q_3$ .
We show that there exists a path $q_1\tr{\beta}{}{\rem\rem}_j'\,\, q_4\tr{\alpha}{}{\rem\rem}_i'\,\, q_3'$
with $q_3 \cong_j q_3'$, or illustratively:
\begin{center}
\newcommand\sigmaold{\sigma}
\def\sigma{q}
%!TEX root = ../main.tex

%RIGHT MOVERS

\def\tatrans{\to}

\begin{tikzpicture}[baseline={([yshift=-.5ex]current bounding box.center)}]
   \tikzstyle{e}=[minimum width=1cm]
   \tikzstyle{every node}=[font=\small, node distance=.5cm]

  \node (s1) {$\sigma_1$};
  \node (s2) [node distance=1.2cm,below of=s1] {$\sigma_2$};
  \node (s3) [right of=s2, xshift=.5cm] {$\sigma_3$};
  \path (s2.east) -- node[pos=.45]{$\stackrel\beta\tatrans_j$} (s3.west);
  \path (s1.south) -- node[midway,sloped]{$\stackrel\alpha\tatrans_i$} (s2.north);

  % grey stuff
  \node (s4) [xshift=.9cm,gray,right of=s1] {$\sigma_4$};
  \node (s3p) [gray,node distance=1.2cm,below right of=s4,xshift=-.5cm] {$\sigma_3'$};
  \path (s1.east) -- node[gray,midway,sloped]{$\stackrel\beta\tatrans_j$}
  (s4.west);
  \path (s4.south) -- node[gray,midway,sloped]{$\stackrel\alpha\tatrans_i$} (s3p);

  % sets
  \node (S1) [left of=s1, xshift=-.2cm,e] {$\LL_i\not\ni$};
  %\node (S2) [left of=s2, xshift=-.3cm,e] {$\RR_i, \RR_j \ni$};
  \node (S2) [below of=s2, xshift=1.7cm,yshift=.3cm,e] {$\in\RR_i$};

  \node (S4) [right of=s4, gray,xshift=.27cm,e] {$\notin\LL_i$};
  \node (S3) [left of=s3, xshift=-1.2cm,e] {$ \RR_i \ni$};

  \path (s3) -- node[gray,sloped,pos=.48]{$\cong_j$} (s3p);

 \node (S3) [gray,right of=s3p, xshift=.27cm,e] {$\in \RR_i$};  
%  
%  \node (s5) [node distance=1.2cm,below right of=s3] {$\sigma_5$};
%  \path (s3) -- node[sloped,pos=.48]{$\tatrans_j$} (s5);
%  \node (S3) [left of=s5, xshift=-.07cm,e] {$\RR_j\ni$};
%  
%  \node (s5p) [gray,node distance=1.2cm,below right of=s3p] {$\sigma_5'$};
%  \path (s5) -- node[gray,sloped,pos=.48]{$\cong_j$} (s5p);
%  \path (s3p) -- node[gray,sloped,pos=.48]{$\tatrans_j$} (s5p);
%  \node (S3) [gray,right of=s5p, xshift=.27cm,e] {$\in \LL_j$};  
%  
%  
\end{tikzpicture}
\end{center}

We have
$M_{\alpha}^\rightarrow(\sigma_1,\alpha,\sigma_2)$ for $\alpha\in\actions_i$
by \autoref{eq:1}, and thus
there is some $B$ such that
$\sst_{\sigma_1}^\rightarrow(B)$, $\alpha \in B$,
$B\cap \en(\sigma_2) \subseteq \actions_i$ and $B\cap\en(\sigma_1) = \set\alpha$ by \autoref{eq:exclude}.
We also have $\sigma_2\tr{\beta}j \sigma_3$ for $j\neq i$ (from $q_2\stackrel{\beta}{\to}{\rem\rem}_j'\,\, q_3$).
As $\beta\in\en(\sigma_2)$, we obtain $\beta\notin B$ by \autoref{eq:exclude} (since $\beta\notin \actions_i$).
Since therefore $\alpha$ right-commutes with $\beta$ by \textbf{\ref{i:d1}},
we obtain $\sigma_1\tr{\beta}j \sigma_4$ and $\sigma_4\tr{\alpha}i \sigma_3$
and according to \autoref{def:trs} also
$q_1\tr{\beta}{}{\rem\rem}_j'\,\, q_4$ and $q_4\tr{\alpha}{}{\rem\rem}_i'\,\, q_3'$
with $q_3' \defn \tuple{\sigma_3, h_3'}$ for some $h_3'$.

Next, we also show that right movability up to $\cong_j$ of~\autoref{i:right} is met,
i.e. $q_3\cong_{j} q_3'$, or $\forall k\neq j \colon h_{3,k} = h_{3,k}'$.
As only transitions $i,j$ are involved, the phases of all other threads
$k\neq i,j$ remain the same according to \autoref{i:invar}.
Furthermore, the transition of $j$ does not influence the phase of $i$ by
\autoref{i:right}, therefore $q_3 \in \RR_i$.
Hence, we only need to show that also $q_3' \in \RR_i$, or $h_{3,i}' = \RR$
(recall that the phase of $j$ may differ according to the definition of $\cong_j$).
%$\cong_{\set{i,j}}$ is not needed.
%Consider again the example from the beginning of the current subsection.

According to \autoref{def:trs}, $q_3'\in \RR_i$ 
iff $q_4 \notin \LL_i \land
M^\rightarrow_i(\sigma_4,\alpha,\sigma_3) \land
\alpha\notin \actions_\ominus^\varphi$. 
We show that all three conjuncts hold:
\begin{enumerate}
\item
As \autoref{def:trs} only allows transitions ending in $\RR_i$ when they
start in $\NN_i$ or $\RR_i$, we have $q_1\notin\LL_i$.
Again, following $j$, we also get $q_4\notin\LL_i$.

\item
\autoref{lem:drm} yields $M^\rightarrow_i(\sigma_4,\alpha,\sigma_3)$ as its premise is assumed
above.

\item $\alpha\notin \actions_\ominus^\varphi$ follows from the initial assumption and \autoref{eq:1}.
\end{enumerate}
For the above, we can conclude that $q_3'\in \RR_i$.
This demonstrates that also $q_3 \cong_j q_3'$, completing this proof.

%Let $\tuple{\sigma_2,\beta,\sigma_3} \in M_\beta^{\to}$, i.e. a transition
%$\beta$ is also  dynamically right moving (in $\sigma_2$), enabled from $\sigma_2$
%and leads to $\sigma_3$. So it again allowed to have $\sigma_3\in\RR_j$.
%The figure shows the right moving operation of $\alpha$ w.r.t $\beta$,
%yielding the grey part of the figure where $\beta$ is executed before $\alpha$
%yielding states $\sigma_3', \sigma_4$.
%
%
%As before, assume that it occurs that
%$\tuple{\sigma_1,\beta,\sigma_4}\notin M_\beta^{\to}$.
%Again, we obtain $\sigma_3\in\RR_j$ but $\sigma_3'\in\LL_j$.
%We see therefore that $\sigma_3',\sigma_4\in\LL_j$ and obtain $\sigma_3 \cong_j \sigma_3'$.
%(e.g. $q_3 = \tuple{\sigma_3,\phase_3}$ and $q_3' = \tuple{\sigma_3',\phase_3'}$
%with $\sigma_3 = \sigma_3'$ but not $\phase_3[j] = \phase_3'[j]$).

%A right mover can exert a similar influence on left movers,
%but the end result is the same (right commutativity up to $\cong_j$).
%As a right-mover can never be disabled again by other threads
%according to  Lemma~\ref{lem:mono}, there are no other cases
%where phases swap and  $\cong_j$ is sufficient to capture their
%behavior (e.g. $\cong_{\set{i,j}}$ is not required).

\item[\phantom{XX}\autoref{i:left}]

%Similar to before,
%\cite[Th.~5]{arxiv} demonstrates that the following condition
%(equivalent to the first case in \cite[Def.~13]{arxiv}) 
%is a sufficient condition for dynamic \emph{left} movability up to $\cong_{\set j}$:
%\[
%(c^\leftarrow_\alpha \lrestr \tr\alpha{i}) \lcomm (c^\leftarrow_\alpha \lrestr \tr\beta{j})
%\,\,\,\land\,\,\,
%c^\leftarrow_\alpha\lrestr \tr{\beta}{j} \rrestr \overline{c^\leftarrow_\alpha}= \emptyset
%\]
%
%Lemma~\ref{lem:mono} fulfils the second conjunct. 
%
%For the first conjunct,
%assume that
%$\sigma_1\in c_\alpha^\leftarrow$ for $\alpha\in\actions_i$.
%Also assume that $\sigma_1\tr{\beta}j \sigma_2$ for $j\neq i$.
%Finally assume that 
%$\tuple{\sigma_2,\alpha,\sigma_3} \in M_{\alpha}^\leftarrow$
%and let $B$ be such that 
%$\sst_{\sigma_1}^\leftarrow(B)\land  \alpha\in B\land B \cap \en(\sigma) \subseteq \en_i(\sigma)$.
%By definition, we have $\sigma_2\in  c_\alpha^\leftarrow$ and 
%$\beta \notin B$
%(as $\beta \in \en(\sigma_1)$ but not $\beta \notin \en_i(\sigma_1)$).
%Therefore, $\beta$ commutes to the right with $\alpha$ with
%$\sigma_1,\sigma_2\in  c_\alpha^\leftarrow$.
%
%Again, we also show that commutativity up to $\cong_{\set{i,j}}$ is sufficient.
%Similar to the dynamic right-mover case, a dynamic left-mover 
%in thread $i$ can indirectly
%cause another thread $j$ to end up in a different phase.
%For example when it moves before a deactivated dynamic left-mover of $j$
%which then dynamically becomes a left-mover.
%We now  show that a left-mover can also influence its own phase.

Assume that $q_2\stackrel{\alpha}{\to}{\rem\rem}_i'\,\, q_3$ with $q_2 \in \LL_i$ and
			$q_1\stackrel{\beta}{\to}{\rem\rem}_j'\,\, q_2$.
We show that there exists a path $q_1\tr{\alpha}{}{\rem\rem}_i'\,\, q_4\tr{\beta}{}{\rem\rem}_j'\,\, q_3'$
with $q_3 \cong_{\set{i,j}} q_3'$, or illustratively:
\begin{center}
\newcommand\sigmaold{\sigma}
\def\sigma{q}
\text{%!TEX root = ../main.tex

\def\tatrans{\to}

\begin{tikzpicture}\centering

   \tikzstyle{e}=[minimum width=1cm]
   \tikzstyle{every node}=[font=\small, node distance=.5cm]

  \node (s1) {$\sigma_1$};
  \node (s2) [right of=s1, xshift=.9cm] {$\sigma_2$};
  \node (s3) [node distance=1.2cm,below right of=s2] {$\sigma_3$};
  \path (s2.south) -- node[midway,sloped]{$\tr{\alpha}i$} (s3.north);
  \path (s1.west) -- node[pos=.45]{$\tr{\beta}j$} (s2.east);

  % grey stuff
  \node (s4) [node distance=1.2cm,gray,below of=s1] {$\sigma_4$};
  \path (s1.south) -- node[gray,midway,sloped]{$\tr{\alpha}i$}
  (s4.north);
  \node (s3p) [gray,node distance=1.2cm,below of=s2] {$\sigma_3'$};
  \path (s4.west) -- node[gray,pos=.48]{$\tr{\beta}j$} (s3p.east);

  % sets
  \node (S1) [left of=s1, xshift=-.2cm,e] {$\LL_i \ni$};
%  \node (S4) [left of=s4, xshift=-.07cm,gray,e] {$\NN_i\ni$};

  \node (S2) [right of=s2, xshift=.2cm,e] {$\in \LL_i$};
%  \node (S3) [right of=s3, xshift=.07cm,e] {$\in \LL_i$};
%  \node (S4) [gray,below of=s3p, xshift=.7cm,e] {$\in \NN_i$};
  \path (s3p) -- node[sloped,pos=.48]{$\cong_{i,j}$} (s3);

%  \node (s5p) [gray,node distance=1.2cm,below right of=s3p] {\vphantom{$\sigma_5'$}}; %phantom
\end{tikzpicture}}
\end{center}

From \autoref{i:invar}, we obtain $q_1 \in \LL_i$.
From the assumption in the \autoref{lem:th}, we get 
$\exists q,q' \colon q \mathmbox{\tr{\alpha'}{}{\rem\rem}_i'}\,\, q' \to'^* q_1$ for some $\alpha'\in\actions_i$.
Without loss of generality, let $q,q'$ be the first on this path,
i.e., that is no $i$-transition on the path $q' \to'^* q_1$.
By \autoref{eq:2}, we obtain
$M_{i}^\leftarrow(\sigma')$ 
and
$\en(\sigma')\cap\actions_i\cap \actions_\oplus^\varphi=\emptyset$.
As that path from $\sigma'$ to $\sigma_2$  (via $\sigma_1$) merely contains transitions from threads $k\neq i$,
we may apply \autoref{lem:dlm} repeatedly to find that 
$M^\leftarrow_i(\sigma_1)$. %Let $B$ now be the semi-$\leftarrow$-stubborn set in $\sigma_1$.
\autoref{eq:left} implies that there is some $B$ such that
$\sst_{\sigma_1}^\leftarrow(B)$ and $B\cap \en(\sigma_1) = \actions_i \cap \en(\sigma_1)$.

Because $\alpha\in \en(\sigma_2) \cap B$, we must have $\alpha\in \en(\sigma_1)$ by \textbf{\ref{i:d1}}.
As also $\beta\notin B$ and $\sst_{\sigma_1}^\leftarrow(B)$, we obtain that $\alpha \lcomm \beta$ in
$\sigma_1$.
Therefore, there are $\sigma_1\tr{\alpha}j \sigma_4$ and $\sigma_4\tr{\beta}i \sigma_3$
and according to \autoref{def:trs} also
$q_1\tr{\alpha}{}{\rem\rem}_j'\,\, q_4$ and $q_4\tr{\beta}{}{\rem\rem}_i'\,\, q_3'$
with $q_3' \defn \tuple{\sigma_3, h_3'}$ for some $h_3'$.
By \autoref{i:invar}, we have that $\forall k \neq i,j \colon h_{3,k}' = h_{3,k}$.
This yields the desired commutativity up to $\cong_{\set{i,j}}$,
completing this proof.
%Apparently, the occurrence of $j$ activated the movability in $i$
%as $i$ becomes a non-mover:

\item[\phantom{XX}\autoref{i:fairness}]
The assumption in \autoref{lem:th} fulfills this requirement immediately.

\item[\phantom{XX}\autoref{i:bisimdisjoint}]
By definition this follows from \autoref{eq:abstract}.

\item[\phantom{XX}\autoref{i:vispre}]
Follows immediately from \autoref{eq:1} in \autoref{def:trs}.

\item[\phantom{XX}\autoref{i:vispost}]
Follows immediately from \autoref{eq:2} in \autoref{def:trs}.

\end{description}
As this covers all the cases in the premise of \autoref{th:vmcai},
we conclude that the lemma holds.
\qed
\end{proof}

%The conclusion of \autoref{th:reduction},
%gives
%$\reach(\cts')\models \Box\varphi \Leftrightarrow \reach(\cts^\brtrans)\models \Box\varphi$.

\setcounter{theorem}{2}
%!TEX root = ../main.tex
\begin{theorem}[Reduced Transaction System (RTS)]\label{th:reduction} We define for all $i$:
\vspace{-2ex}
\begin{IEEEeqnarray*}{llr}
%\NN_i &\defn \set{ \tuple{\sigma,\phase} \mid \phase_i = \textsf{\NN} }, \text{etc} &
%\text{\hspace{-2cm}(state partitioning: $\overline{\NN_i} = \LL_i\cup\RR_i$)}\\
%Let $q \overarrow{\alpha}_i' q'$ for $\tuple{q,\alpha,q'}\in \ttrans_i$.\\
%$W \defn \bigcup_i W_i$, $N \defn \bigcap_i N_i$ and 
%$N_\tx \defn \bigcap_{u\neq t} N_j$.\\
\trtrans_i &\defn (\cup_{j\neq i} \NN_j) \lrestr \trp
					&\text{($i$  only transits when all $j$ are external)}\\
\brtrans_i &\defn \NN_i\lrestr (\trtrans_i\rrestr \overline{\NN_i})^* \trtrans_i \rrestr \NN_i
&\text{(skip internal states transition relation)}%\\
%\brtrans &\defn \bigcup_i \brtrans_i  &\\
%\stackrel\brtrans{\scriptsize \cts} &\defn \langle \tstates,\{\tuple{q,\alpha_i,q'}\mid
%                                        q\stackrel\alpha\brtrans_i q' \},\actions,\tsint \rangle)
%                                        &\text{(the reduced transaction system)}
\vspace{-2ex}
\end{IEEEeqnarray*}
The RTS is a CST
$\stackrel\brtrans{\scriptsize \cts} \defn\rem\rem\rem \langle \tstates,\{\tuple{q,\alpha_i,q'}\mid
                                        q\stackrel{\alpha_i}\brtrans_i q' \},\actions,\tsint \rangle$.
Now, provided that
$\forall \sigma \in \LL_i \colon \exists \sigma'\in\NN_i \colon \sigma\hookrightarrow^*_i \sigma' $,
we have
$\reach(\cts')\models \Box\varphi \Longleftrightarrow
\reach(\stackrel\brtrans{\scriptsize \cts})\models \Box\varphi$.
\end{theorem}

\begin{proof}
\autoref{lem:th} asummes termination
$\forall \sigma \in \LL_i \colon \exists \sigma'\in\NN_i \colon \sigma\hookrightarrow^*_i \sigma' $
and actuation
$\forall \sigma \in \LL_i \colon \exists \sigma',\sigma''  \colon \sigma' \hookrightarrow_i \sigma'' \hookrightarrow^* \sigma$
of post phases.
The termination assumption is met by the `provided' assumption of the theorem. 
The actuation is met by the theorem's use
of $\brtrans$, which only considers the subsystem of full transactions starting
with and ending in external states:
$ \brtrans_i \defn \NN_i\lrestr (\trtrans_i\rrestr \overline{\NN_i})^* \trtrans_i \rrestr \NN_i$.
As the premise of \autoref{lem:th} is therefore met,
it follows that we can apply \autoref{th:vmcai}.
Therefore,
if $\sigma\to'^{*} \sigma'$ with $\sigma \in\bigcap_i \NN_i \cap \varphi$ and $\sigma'\in \overline\varphi$, then
$\exists\sigma''\in \overline\varphi$ s.t. $\sigma \,\brtrans^{*} \,\sigma''$. 
As therefore invariant violations are preserved by the reduction, we have
$ \reach(\cts')\not\models \Box\varphi \implies 
\reach(\stackrel\brtrans{\scriptsize \cts}) \not\models \Box\varphi$.
Because $\trtrans_i \subseteq \to_i'$ and $\brtrans_i^* \subseteq \to_i'^*$,
we also have the opposite
$\reach(\stackrel\brtrans{\scriptsize \cts}) \not\models \Box\varphi 
\implies \reach(\cts')\not\models \Box\varphi$, i.e.
as by definition the reduced system 
must contain a subset of the invariant violations in the full system.
Taken together (conjoining the contrapositions),
\autoref{th:reduction} is satisfied.
\qed
\end{proof}

%\begin{theorem}\label{lem:th}
%Let 
%$\cts^\brtrans \defn \tuple{\tstates,\set{\tuple{q,\alpha,q'}\mid
%                                        q\stackrel\alpha\brtrans q' },\actions,\tsint})$.
%Taking $Y = \varphi\times \phases$, we obtain
%$\reach(\cts')\models \Box\varphi \Longleftrightarrow \reach(\cts^\brtrans)\models \Box\varphi$.
%\end{theorem}

\begin{theorem}\label{th:alg}
\autoref{alg:rtrs} computes $\reach(\brcts)$ s.t.
$\reach(\cts)\models \Box\varphi \Longleftrightarrow
\reach(\brcts)\models \Box\varphi$.
%$\forall \sigma \in \LL_i \colon \exists \sigma'\in\NN_i \colon \sigma\hookrightarrow^*_i \sigma' $.
\end{theorem}

\begin{proof}
The algorithm uses the approach described by Valmari~\cite{valmari1} to
ensure that 
$\forall \sigma \in \LL_i \colon \exists \sigma'\in\NN_i \colon \sigma\hookrightarrow^*_i \sigma'$.
It follows therefore that the premise of \autoref{lem:th} holds (including the ``provided that'' part).
From \autoref{lem:preserves}, we also have
$\reach(\cts)\models \Box\varphi \Longleftrightarrow
\reach(\cts')\models \Box\varphi$,
hence:
$\reach(\cts)\models \Box\varphi \Longleftrightarrow
\reach(\brcts)\models \Box\varphi$.
\qed
\end{proof}

\begin{theorem}\label{th:removal}
Let $N\defn \cap_{i} \NN_i$.
We have $\sizeof{N} = \sizeof{S}$ and
$\reach(\brcts)\subseteq \reach(\cts)$.
\end{theorem}
\begin{proof}
By definition of the initial state $\sigma_0'$ in \autoref{def:trs}, the reduced transition relations
$\trtrans$, $\brtrans$ in \autoref{th:reduction} and basic induction.
\qed
\end{proof}

\end{document}